\documentclass{birkjour_t2}
\usepackage[latin1]{inputenc}
\usepackage{amsmath}
\usepackage{amsfonts}
\usepackage{amssymb}
\usepackage{amsthm}
\usepackage{hyperref}
\usepackage{graphicx}
\usepackage[ruled,vlined,linesnumbered]{algorithm2e}
\usepackage{color}
\usepackage{bm}
\usepackage{tikz}
\usepackage{enumerate}

\theoremstyle{plain}
\newtheorem{theorem}{Theorem}

\newtheorem{lemma}{Lemma}
\theoremstyle{definition}
\newtheorem{definition}{Definition}
\newtheorem{example}{Example}
\theoremstyle{remark}
\newtheorem{remark}{Remark}

\newcommand{\scad}{sub-CAD}
\newcommand{\scads}{sub-CADs}
\newcommand{\ttiscad}{sub-TTICAD}

\newcommand{\Scad}{Sub-CAD}

\def\softO{\tilde O}

\newcommand{\llaymanscad}[1]{#1-layered variety sub-CAD}
\newcommand{\LayManScad}{Layered Variety Sub-CAD}
\newcommand{\lLayManScad}[1]{#1-Layered Variety Sub-CAD}
\newcommand{\LMscad}{LV-sub-CAD}
\newcommand{\lLMscad}[1]{#1-LV-sub-CAD}
\newcommand{\laymanttiscad}{layered variety sub-TTICAD}
\newcommand{\llaymanttiscad}[1]{#1-layered variety sub-TTICAD}
\newcommand{\LayManttiScad}{Layered Variety Sub-TTICAD}

\newcommand{\LMttiscad}{LV-sub-TTICAD}

\definecolor{orange}{rgb}{1,0.5,0}
\newcommand{\TODO}[1]{\colorbox{orange}{{\bf #1}}}

\author{D.~J.~Wilson, R.~J.~Bradford, J.~H.~Davenport and M.~England}
\title[Cylindrical Algebraic Sub-Decompositions]{Cylindrical Algebraic Sub-Decompositions}
\address{Department of Computer Science, University of Bath, Bath, BA2 7AY, England}
\email{\{D.J.Wilson, R.J.Bradford, J.H.Davenport, M.England\}@bath.ac.uk}

\begin{document}

\begin{abstract}
Cylindrical algebraic decompositions (CADs) are a key tool in real algebraic geometry, used primarily for eliminating quantifiers over the reals and studying semi-algebraic sets.  In this paper we introduce cylindrical algebraic sub-decompositions (sub-CADs), which are subsets of CADs containing all the information needed to specify a solution for a given problem. 

We define two new types of sub-CAD: variety sub-CADs which are those cells in a CAD lying on a designated variety; and layered sub-CADs which have only those cells of dimension higher than a specified value.  We present algorithms to produce these and describe how the two approaches may be combined with each other and the recent theory of truth-table invariant CAD.  

We give a complexity analysis showing that these techniques can offer substantial theoretical savings, which is supported by experimentation using an implementation in {\sc Maple}.
\end{abstract} 
%Abstract 183 words (guideline 150-250)

%Guideline 4-6 keywords

\maketitle

\section{Introduction}

\subsection{Motivation}

A {\bf cylindrical algebraic decomposition} (CAD) is a decomposition of $\mathbb{R}^n$ into cells arranged cylindrically (meaning the projections of any pair of cells onto the first $k$ coordinates are either equal or disjoint) each of which is a semi-algebraic set (and so may be described by polynomial relations).  
They are traditionally produced sign-invariant with respect to a list of polynomials meaning each polynomial has constant sign on each cell.  
CAD was introduced by Collins in \cite{Collins1975}, and has become a key tool in real algebraic geometry for studying semi-algebraic sets and eliminating quantifiers over the reals.  Other applications include robot motion planning \cite{SS83II}, parametric optimisation \cite{FPM05}, epidemic modelling \cite{BENW06}, theorem proving \cite{Paulson2012} and programming with complex functions \cite{DBEW12}.

CAD usually produces far more information than required to solve the underlying problem. Often a problem will be represented by a formula and we thus require a CAD such that the formula has constant truth value on each cell.  This can be achieved by building a CAD sign-invariant for the polynomials in the formula, but that may introduce cell divisions not relevant to the formula itself.  
Many techniques have been developed to try and mitigate this, some of which we discuss later.  However, even then algorithms may produce thousands of superfluous cells which are not part of the solution set.  The key focus of this paper is the development of methods to return a subset of a CAD sufficient to solve a given problem.  We show that such subsets can often be identified from the structure of the problem, motivating the following new definitions.

\begin{definition}{\ }\\
\label{def:SubCAD}
Let $\mathcal{D}$ be a CAD of $\mathbb{R}^n$ (represented as a set of cells). Then a subset $\mathcal{E} \subseteq \mathcal{D}$ is a {\bf cylindrical algebraic sub-decomposition ({\scad})}. 

Let $F \subset \mathbb{Q}[x_1,\ldots,x_n]$. If $\mathcal{D}$ is a sign-invariant CAD for $F$ then $\mathcal{E}$ is a {\bf sign-invariant {\scad}}.  We define {\scad}s with other invariance properties in an analogous manner, such as \textbf{truth-invariance} for a Tarski formula $\varphi(x_1,\ldots,x_n)$.  If $\mathcal{D}$ is a truth-invariant CAD for $\varphi$ and $\mathcal{E}$ contains all cells of where $\varphi$ is satisfied then we say that $\mathcal{E}$ is a {\bf $\bm{\varphi}$-sufficient {\scad}}.

%Assume each cell in $\mathcal{D}$ and $\mathcal{E}$ has a cell index. Then if for every cell in $\mathcal{E}$, its index is the same in both $\mathcal{E}$ and $\mathcal{D}$, then we say $\mathcal{E}$ is an {\bf index consistent {\scad}} of $\mathcal{D}$.
\end{definition}

As an example of when sub-CADs may be applicable, consider quantifier elimination, the original motivation for CAD.  Given a quantified formula $\varphi$ we want to derive an equivalent quantifier-free formula.  For a formula over the reals this is achieved by constructing a sign-invariant CAD for the polynomials in $\varphi$ and testing the truth of $\varphi$ at a sample point of each cell.  This is sufficient to draw a conclusion for the whole cell due to sign-invariance and thus an equivalent quantifier free formula can be created from the algebraic description of the cells on which $\varphi$ is true.  Such an application makes no use of the cells on which $\varphi$ is false and so a ${\varphi}$-sufficient {\scad} is appropriate.  

Of course, for a given problem we would like the smallest possible $\varphi$-sufficient {\scad}.  It is not usually possible to pre-identify this, but we have developed techniques which restrict the output of the CAD algorithm to provide {\scad}s sufficient for certain general classes of problems.  These will offer savings on any subsequent computations on the cells (such as evaluating polynomials or formulae) and in some cases also offer substantial savings in the CAD construction itself. 
We will introduce these techniques and demonstrate how they can be combined with each other and additional existing CAD theory, but first remind the reader of the necessary background theory.

\subsection{Background to CAD}

Collins' original algorithm is described in \cite{ACM84I}.  While there have been many improvements and refinements to this algorithm the structure has remained largely the same.  In the first phase, {\bf projection}, a {\bf projection operator} is repeatedly applied to a set of polynomials, each time producing another set in one fewer variables.  Together these sets contain the {\bf projection polynomials}.  These are then used in the second phase, {\bf lifting}, to build the CAD incrementally.  First $\mathbb{R}$ is decomposed into cells: points corresponding to the real roots of the univariate polynomials, and the open intervals defined by them.  Then $\mathbb{R}^2$ is decomposed by repeating this process over each cell using the bivariate polynomials (evaluated at a sample point).  The output for each cell consists of {\bf sections} (where a polynomial vanishes) and {\bf sectors} (the regions between). Together these form a {\bf stack} over the cell, and taking the union of these stacks gives the CAD of $\mathbb{R}^2$.  This is repeated until a CAD of $\mathbb{R}^n$ is produced.     
To conclude that the CAD is sign-invariant we need delineability.  A polynomial is {\bf delineable\/} over a cell if the portion of its zero set over that cell consists of disjoint sections.  Then a set of polynomials is {\bf delineable\/} over a cell if each is delineable and the sections of different polynomials over the cell are either identical or disjoint.  The projection operator used must ensure that over each cell of a sign-invariant CAD for the projection polynomials in $r$ variables, the polynomials in $r+1$ variables are delineable.

All cells include a cell index and a sample point.  The index is an $n$-tuple of positive integers that corresponds to the location of the cell relative to the rest of the CAD. Cells are numbered in each stack during the lifting stage (from most negative to most positive), with sectors having odd numbers and sections having even numbers.  Therefore the dimension of a given cell can be easily determined from its index: simply the number of odd indices in the $n$-tuple. Our algorithms in this paper will produce {\scads} that are {\bf index-consistent}, meaning a cell in a {\scad} will have the same index as it would in the full CAD.  Further, we will assume that cells are stored lexicographically by index.

Important developments to CAD include: refinements to the projection operator \cite{Hong1990, McCallum1998, Brown2001a}, reducing the number of projection polynomials and hence cells; 
%algorithms to identify the adjacency of cells in a CAD \cite{ACM84II, ACM88} and following from this the idea of clustering \cite{Arnon1988} to minimise the lifting.
partial CAD \cite{CH91}, where the structure of the input formula is used to simplify the lifting stage; the theories of equational constraints and truth-table invariance \cite{McCallum1999, BDEMW13} where the presence of equalities in the input further refines the projection operator; the use of certified numerics in the lifting phase \cite{Strzebonski06, IYAY09}; and CAD via triangular decomposition \cite{CMXY09} which constructs a decomposition of complex space and refines this to a CAD.

Constructing a CAD is doubly exponential in the number of variables \cite{DH88}.  While none of the improvements described above (or introduced in this paper) circumvent this they do make a great impact on the practicality of using CAD.  Note that CAD can depend heavily on the variable ordering used (from linear to doubly-exponential \cite{BD07}).  In this paper we work with polynomials in $\mathbb{Q}[\bm{x}]$ with the variables ${\bf x} = x_1,\ldots,x_n$ in ascending order (so we first project with respect to $x_n$ and continue until we reach  univariate polynomials in $x_1$).  The {\bf main variable} of a polynomial (${\rm mvar}$) is the greatest variable present with respect to the ordering.
Heuristics to assist with selecting a variable ordering (and other choices) were discussed in \cite{DSS04, BDEW13} and are equally applicable to {\scad}s.

\subsection{New Contributions}

In Section \ref{sec:scad} we present new algorithms to produce {\scads}, as well as surveying the literature to identify other examples of {\scads}.  To the best of our knowledge the concept of a sub-CAD has never been formalised and unified before. 

We start in Section \ref{sec:MCAD} by defining a {\bf Variety {\scad} (V-{\scad})}. This idea combines the ideas of: equational constraints \cite{McCallum1999}, where the presence of an equation implied by the input formula improves the projection operator; and partial CAD \cite{CH91}, where the logical structure of the input allows one to truncate the lifting process when the truth value can already be ascertained.  We observe that if the input formula contains an equational constraint then all valid cells must lie on the variety it defines and hence it is unnecessary to produce cells not on this variety.  
%We give an algorithm to produce a V-{\scad} and show that it can offer substantial savings over a standard CAD algorithm.

In Section \ref{sec:LCAD} we define a  {\bf Layered {\scad} (L-{\scad})} as the cells in a CAD of a specific dimension or higher.  It has been noted previously that a problem involving only strict inequalities would require only the cells in a CAD of full-dimension \cite{McCallum1993,Strzebonski00}.  We generalise this idea and explain when it may be of use, for example to solve problems whose solution sets are of known dimension or in applications like robot motion planning where only cells of certain dimensions are of use.

These new ideas improve the practicality of using CAD and their effect can be increased by combining them, as discussed in Section \ref{subsec:LMCAD}.
For example, consider formulae of the form $f = 0 \land \varphi$ where $\varphi$ involves only strict inequalities. Then a {\bf {\LayManScad} ({\LMscad})} can provide the cells of full dimension on the variety and thus the generic families of solutions.

These new ideas may also be combined with many existing aspects of CAD theory.  It is of course sensible to combine the restricted output of a variety CAD with the theory of reduced projection with respect to an equational constraint.  However, it is also possible to combine with the projection operator for {\bf truth-table invariant CAD (TTICAD)} recently presented in \cite{BDEMW13}.  A TTICAD is one for which each cell is truth invariant for a list of formulae, utilising equational constraints in the individual formulae to reduce the number of projection polynomials.  We discuss when and how truth-table-invariant {\scad}s can be produced in Section \ref{subsec:LMTTICAD}.
%We explain how to combine to form {\bf Variety {\ttiScad} (V-{\ttiscad}s)}, {\bf Layered {\ttiScad}s (L-{\ttiscad}s)}  or even {\bf  {\LayManttiScad}s ({\LMttiscad}s)}.  
In Section \ref{SUBSEC:LMTTICASD} we examine a problem where a {\LMttiscad} can be used to identify almost all the solutions (the set of missing solutions has measure zero). This approach produces 88\% fewer cells than using TTICAD alone and takes seconds rather than minutes (while trying to tackle the problem with a traditional CAD is infeasible).  This and two other case studies demonstrating the benefit of the new algorithms are presented in Section \ref{sec:ExImp}.

In Section \ref{sec:Complexity} we give a complexity analysis of certain {\scads}.  
Although none of the new theory allows us to avoid the doubly exponential nature of CAD they do allow for improved asymptotic bounds.  The improvement is a drop in the constant term of the double exponent, and we note that $2^{2^{n}} \neq O(2^{2^{n-1}})$ so such savings can have a substantial effect.  This is reflected by experimental results in Section \ref{sec:ExImp} where substantial increases in efficiency due to sub-CAD technology are demonstrated.

\section{Sub-CADs}
\label{sec:scad}

We aim to return only those cells necessary to solve the problem at hand: a $\phi$-sufficient {\scad}.  However, trying to identify the minimal $\phi$-sufficient {\scad} for a problem would mean essentially solving the problem itself and so we instead explain how to identify sets of valid or invalid cells during the lifting stage at minimal cost.  Indeed, the two new approaches to {\scad} we present require only simple checks on cell-dimensions (easily obtained via the cell-index). %Building {\scads} this way saves subsequent work on unimportant cells (such as evaluation of polynomials) and can also give savings in CAD construction time.
We present theory and algorithms for variety and layered {\scads} in Sections \ref{sec:MCAD} and \ref{sec:LCAD}, and then in Section \ref{subsec:furtherscads} we put our work into context by surveying the CAD literature for relevance to {\scads}.  

\subsection{Variety {\scads}}
\label{sec:MCAD}

Recall the definition of an equational constraint.
\begin{definition}{\ }\\
\label{def:EC}
Let $\varphi$ be a Tarski formula. An {\bf equational constraint} is an equation, $f=0$, logically implied by $\varphi$.
\end{definition}
Equational constraints may be given explicitly (as in $f=0 \land \phi$), or implicitly (as $f_1f_2=0$ is in $(f_1 = 0 \land \phi_1) \lor (f_2 = 0 \land \phi_2)$).  The presence of an equational constraint can be utilised in the first projection stage by refining the projection operator \cite{McCallum1999} and also in the final lifting stage by reducing the amount of polynomials used to construct the stacks \cite{England13b}.  If more than one equational constraint is present then further savings may be possible \cite{McCallum2001, BM05}. We restrict ourselves to a single equational constraint, and if multiple equational constraints are present we assume that one has been designated.
\begin{definition}{\ }\\
\label{def:Vscad}
Let $\varphi$ be a Tarski formula with equational constraint $f=0$.  A truth-invariant {\scad} for $\varphi$ consisting only of cells lying in the variety defined by $f = 0$ is a {\bf Variety {\Scad} (V-{\scad})}.   
\end{definition}
Partial CAD \cite{CH91} describes how the logical structure of the input formula is used to truncate lifting when possible in CAD.  For example, if the truth of an expression on a cell $c$ can already be determined then there is no need to lift over it.  Algorithm \ref{alg:VarietySub-CAD} combines the ideas of utilising equational constraints and partial lifting to build variety {\scads} in the case where all factors of the equational constraint have the main variable of the system.

In Algorithm \ref{alg:VarietySub-CAD} $A$ is a square-free basis for the polynomials defining $\varphi$ and $E$ the subset of those defining $f$.  {\tt ProjOp} refers to an algorithm implementing a suitable CAD projection operator.  Then {\tt CADAlgo} and {\tt GenerateStack} respectively implement compatible algorithms for CAD construction, and stack generation over a cell with respect to the sign of given polynomials.  By compatible we mean using the same projection operator and checking for any necessary conditions of its use.  This is required as some CAD algorithms may return FAIL if the input does not satisfy certain conditions.  These are checked for during stack construction and usually referred to as the input being {\bf well-oriented} (see for example \cite{McCallum1998}).  In these cases Algorithm \ref{alg:VarietySub-CAD} must also return FAIL.
A sensible choice of projection operator is one which utilises the equational constraint to minimise the number of projection polynomials, such as $P_E(A)$ from \cite{McCallum1999}.  We verify the correctness of Algorithm \ref{alg:VarietySub-CAD} for this choice.

\begin{algorithm}[ht]
%\DontPrintSemicolon

\SetKwInOut{Input}{Input}\SetKwInOut{Output}{Output}

\Input{A formula $\varphi$, a declared equational constraint $f=0$ from $\varphi$ and variables $\bm{x}=x_1, \dots, x_n$.  $\varphi$ is in $\bm{x}$ and all factors of $f$ have main variable $x_n$.
}
\Output{A (truth-invariant) variety {\scad} of $\mathbb{R}^n$ for $(\varphi, f)$, or FAIL.}
\BlankLine

Extract from $\varphi$ the set of polynomials $A$ and from $f$ the subset $E \subset A$\;

${\bf P} \leftarrow$ output from applying ${\tt ProjOp}$ to (A,E) once
\tcp*{First projection stage}

$\mathcal{D}' \leftarrow {\tt CADAlgo}({\bf P}, [x_1,\ldots,x_{n-1}])$
\tcp*{Computation of a CAD of $\mathbb{R}^{n-1}$}

\If{ $\mathcal{D}'$ = {\rm FAIL}}{
\Return {\rm FAIL}  \tcp*{${\bf P}$ is not well oriented}
}

$\mathcal{D} \leftarrow []$\;

\For{$c \in \mathcal{D}'$}{
  $S \leftarrow {\tt GenerateStack}(E, c)$\tcp*{Final lifting stage}
  \label{step:finallift}
  \If{ $S$ = {\rm FAIL}}{
  \Return {\rm FAIL}  \tcp*{Input is not well oriented}
  }
  \If{$|S| > 1$}{
    \For{$i=1 \dots (|S|-1)/2$}{
      $\mathcal{D}.{\tt append}(S[2i])$ \tcp*{Cells with even index included}
    }
  }
}

\Return $\mathcal{D}$\;

\caption{${\tt VarietySubCAD}(\varphi,f,{\bf x})$: Algorithm to produce variety {\scad}s.}
\label{alg:VarietySub-CAD}
\end{algorithm}

\begin{theorem}
\label{thm:MCAD}
When the sub-algorithms are chosen to implement McCallum's algorithm to produce CADs with respect to an equational constraint \cite{McCallum1999}, then Algorithm \ref{alg:VarietySub-CAD} satisfies its specification, with the outputted {\scad} consisting of cells on which the input formula has constant truth value.
\end{theorem}

\begin{proof}
The algorithm in \cite{McCallum1999} applies a projection operator $P_E(A)$ and then incrementally constructs CADs of increasing real dimension, checking for well-orientedness when building each stack. In \cite{McCallum1999} the authors proved that the CAD returned was truth-invariant for the equational constraint and sign-invariant for any other polynomials involved on cells where the equational constraint was satisfied. 

The first difference in Algorithm \ref{alg:VarietySub-CAD} is in step \ref{step:finallift} where the final lift is performed with respect to $E$ rather than $A$.  In fact, this improvement follows directly from Theorem 2.2 in \cite{McCallum1999}, although it was not realised until \cite{BDEMW14}.  This reduces the size of the output, but not its invariance structure or correctness.

Next, in the final loop, only some of the cells generated in the final lift are included in the output.  The cells in question are a subset of what would have been produced otherwise and thus certainly a {\scad} with the same invariance property.  It remains to prove that they are a variety {\scad} (in which case we can conclude $\varphi$ has constant truth value on each cell).  

The selected cells are those with even index (the sections).  If the polynomial $f$ is not identically zero over a cell in $\mathcal{D}$ then these must together define its variety.  If any of the polynomials in $E$ were nullified then part of the variety may be in the sectors, but in this case the input would have failed the well-orientedness condition in \cite{McCallum1999} and thus Algorithm \ref{alg:VarietySub-CAD} would return FAIL.
\end{proof}

\begin{remark}
As the operator from \cite{McCallum1999} can return FAIL in situations where others do not, we now consider how Algorithm \ref{alg:VarietySub-CAD} may be adapted to use alternative CAD projection operators.

First, if a polynomial in $E$ is nullified on a cell of $\mathcal{D}$ then \cite{McCallum1999} returns FAIL while McCallum's operator to produce sign-invariant CADs in \cite{McCallum1998} is still applicable (because then the nullification is in the final lift where only sign-invariance and not order invariance is required).  However, we cannot simply apply Algorithm \ref{alg:VarietySub-CAD} with the alternative projection operator as it will now be the case that some of the variety is contained in the sectors over the cell in question.  In this case we would need the \texttt{GenerateStack} algorithm to check for nullification of $E$, and then if it occurs have Algorithm \ref{alg:VarietySub-CAD} include all cells from that stack in the output.  

Second, if some other polynomial is nullified causing failure then the original algorithm of Collins (or its improvement by Hong \cite{Hong1990}) is still applicable.  As with the previous case we must still check for nullification over a cell in the final lift, including the whole stack when nullification occurs.

In these cases there would still be output savings from building a variety {\scad}  since the inclusion of the full stack only needs to happen over those cells where nullification occurs.
\end{remark}

We demonstrate the savings in output size offered with a simple example.
\begin{example}{\ }\\
\label{ex:MCAD}
Assume variable ordering $x \prec y$ and define the polynomials 
\[
f := x^2 + y^2 - 1, \qquad g := x
\]
which are graphed with solid curves in each of the images in Figure \ref{fig:MCAD}.  This first of these images visualises the simplest sign-invariant CAD for the polynomials, which has 23 cells (each indicated by a solid box).  If the box is at the intersection of two curves (including the dotted lines) then the cell it represents is a point.  Otherwise, if the box is on a curve then the cell represented is that portion of the curve and if the box is not on a curve then the cell represented is that portion of the plane.

\begin{figure}[ht]
\begin{center}
\includegraphics[width=0.45\textwidth]{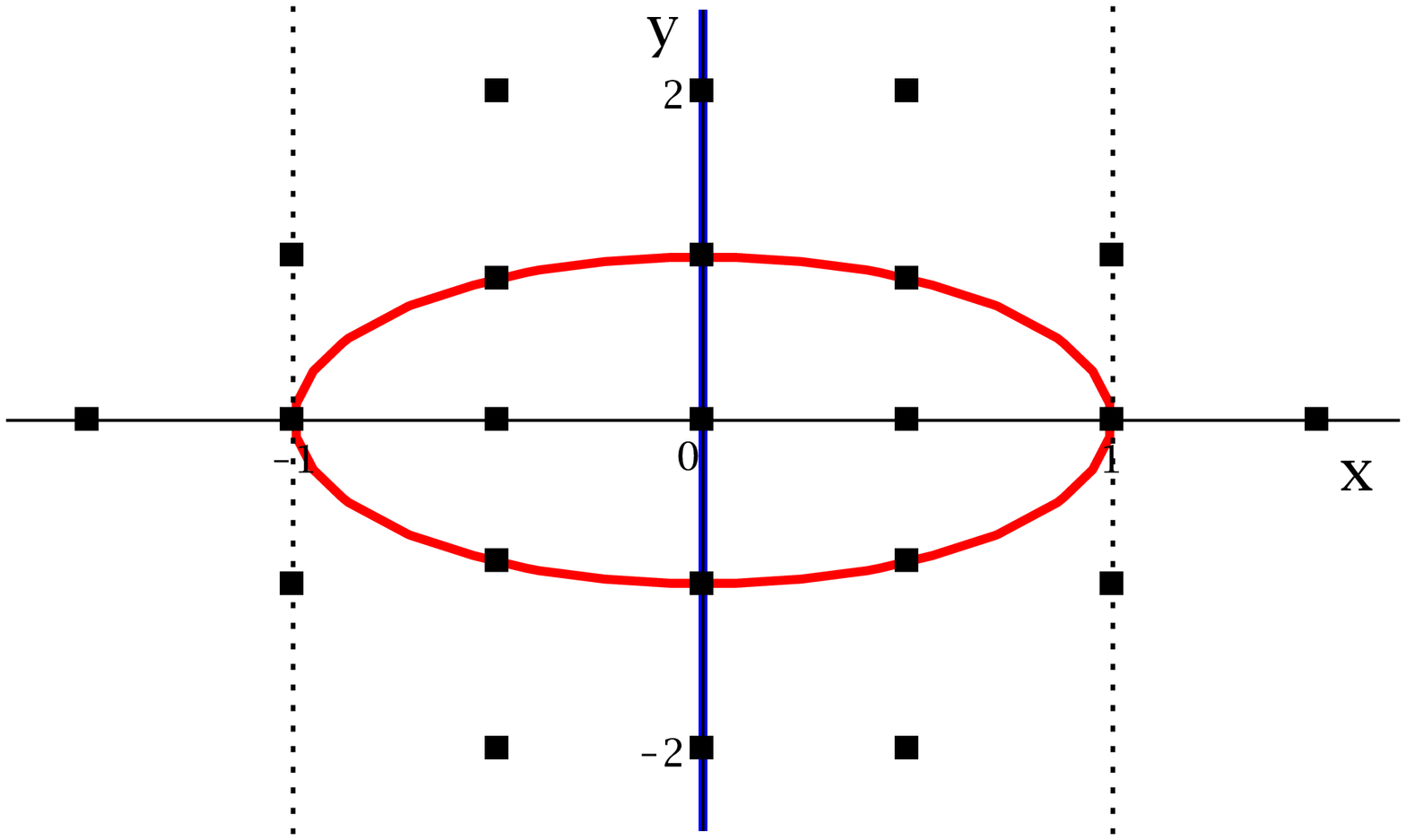}
\includegraphics[width=0.45\textwidth]{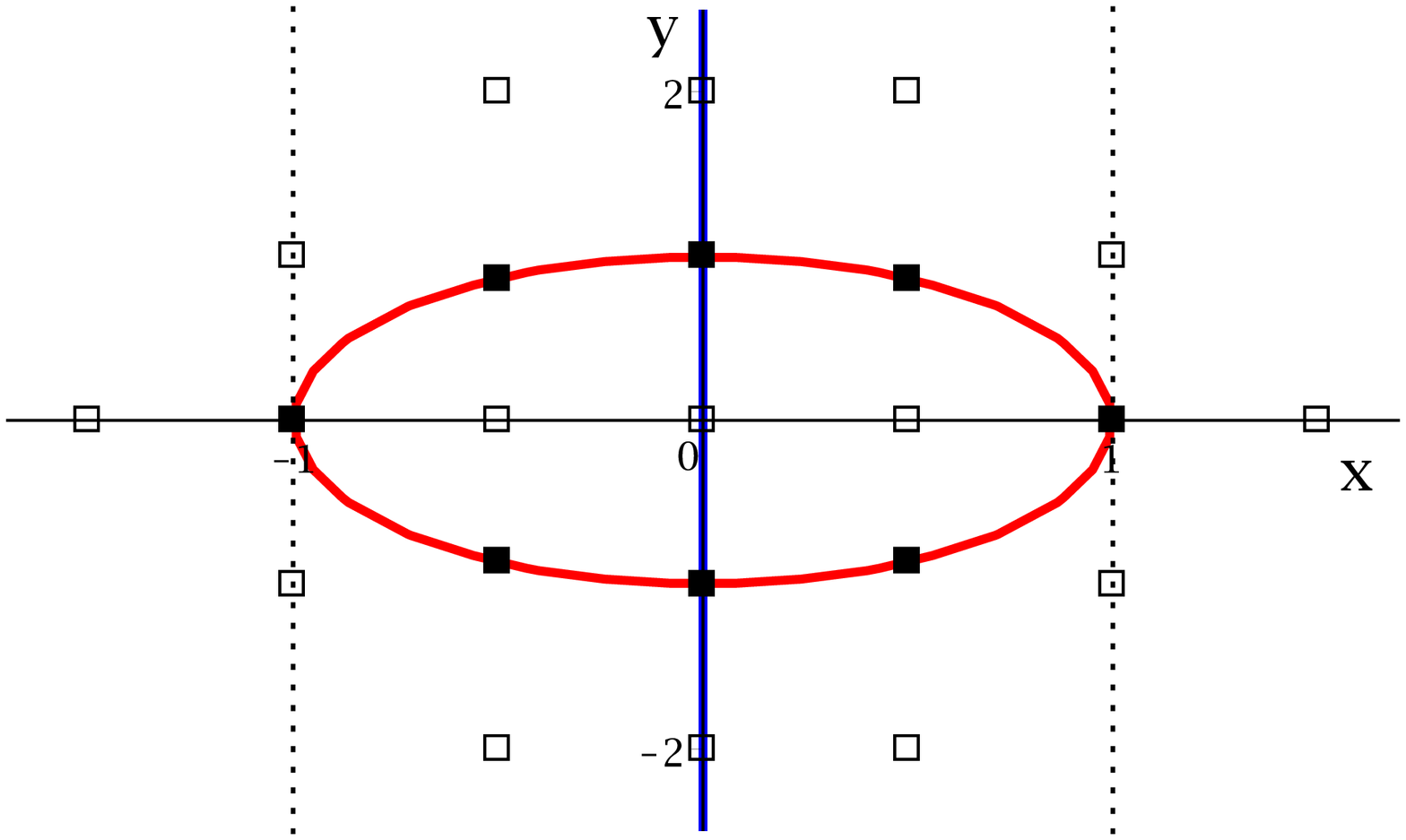}
\end{center}
\caption{Images representing the CADs described in Example \ref{ex:MCAD}.  The solid boxes represent cells in $\mathbb{R}^2$, (with the empty boxes cells that were constructed but discarded).}
\label{fig:MCAD}
\end{figure}

Suppose these polynomials originated from a problem involving the formula $\varphi_1 := f=0 \land g<0$.  Then 3 of the cells describe the solution (those on the circle to the left of the $y$-axis).  Using Algorithm \ref{alg:VarietySub-CAD} a V-{\scad} would be returned with all 8 of the cells on the circle.  The other cells have been recognised as not being on the variety to which all solutions belong and hence have been discarded.  The  image on the right in Figure \ref{fig:MCAD} identifies the cells returned in the {\scad} with solid boxes.  
\end{example}

The cell counts in this example were modest, but we will see later in Section \ref{sec:ExImp} that for more complicated examples the savings offered by using a V-{\scad} can be substantial.  Using Algorithm \ref{alg:VarietySub-CAD} clearly reduces the output size of CAD but it will save little CAD computation time since most of the work required to define the discarded cells (such as root isolation) had to be performed. These computed but discarded cells are shown as empty boxes in Figure \ref{fig:MCAD}.  
There would, however, be savings in computation time for any further work, such as only having to evaluate polynomials on 8 cells instead of 23 in order to describe the solutions.

If a formula has an equational constraint with factors not in the main variable of the system (parts of the corresponding variety have lower dimension) then building a V-{\scad} would mean even more potential cell savings, but also savings in computation time since some stacks would never be built at all.  
Suppose each factor of the equational constraint has main variable $x_k$ where $k<n$.  To adapt Algorithm \ref{alg:VarietySub-CAD} we must perform the restriction when lifting to a CAD of $\mathbb{R}^k$, instead of $\mathbb{R}^n$.  Thus we would first build a CAD of $\mathbb{R}_{k-1}$, then perform the restricted lifting to a {\scad} of $\mathbb{R}^k$, and continue lifting to a {\scad} of $\mathbb{R}^n$.  However, verifying this approach is a little more subtle.  We cannot follow Theorem \ref{thm:MCAD} and use McCallum's reduced projection at the start (as the set $E$ is empty \cite{McCallum1999}) but if we  were to use it to build a {\scad} of $\mathbb{R}^k$ then we would need to take care in how we lift over it (ensuring the projection polynomials above are delineable).  We see two possibilities:
\begin{enumerate}[(a)]

\item Use the tools of McCallum's sign-invariant algorithm from \cite{McCallum1998} throughout.  In particular, we must perform the restricted lifting with respect to the full set of projection polynomials of main variable $x_k$ rather than just those defining the equational constraint.  This is because to continue lifting with respect to projection polynomials provided by McCallum's operators we need to conclude that the {\scad} of $\mathbb{R}^k$ is order invariant on the cells, rather than just sign-invariant.  Thus the outputted {\scad} is not a variety {\scad} but a superset of cells containing one.  Algorithm \ref{alg:V2} demonstrates this approach.

\item Use Collins-Hong projection \cite{Hong1990} for the first $(n-k)$ projection stages.  Then apply Algorithm \ref{alg:VarietySub-CAD} with McCallum's reduced projection operator for equational constraints to build a variety {\scad} of $\mathbb{R}^k$ (as verified by Theorem \ref{thm:MCAD}) before continuing lifting to a variety {\scad} of $\mathbb{R}^n$.
\end{enumerate}
The latter approach is still a variety {\scad} and allows for a smaller CAD of $\mathbb{R}^k$ but these benefits may be overshadowed in the final {\scad} of $\mathbb{R}^n$ due to lifting with respect to a larger set of polynomial in the later stages.  Of course both cases may return FAIL in which case having all sub-algorithms implement \cite{Hong1990} would be the only way forward.  However, it is worth noting that when the equational constraint is not in the main variable we may actually avoid unnecessary failure:  if the input is not well-oriented but the problematic nullification only occurs on cells that are not on the variety then the outputted {\scad} will still be valid. It is interesting to note that this {\scad} is now a subset of a CAD we do not know how to produce algorithmically.  

\begin{algorithm}[ht]
%\DontPrintSemicolon

\SetKwInOut{Input}{Input}\SetKwInOut{Output}{Output}

\Input{A formula $\varphi$, a declared equational constraint $f=0$ from $\varphi$ and variables $\bm{x} = x_1, \dots, x_n$.  $\varphi$ is in $\bm{x}$ and all factors of $f$ have the same main variable.
}
\Output{A {\scad} $\mathcal{D}$ on which $\varphi$ is truth invariant and which is the superset of a variety {\scad} for $(\varphi, f)$, or FAIL.}
\BlankLine

Extract from $\varphi$ the set of polynomials $A$\;
% and from $f$ the subset $E \subset A$ \;

Set $k$ to be the index of mvar$(f)$ \;

Perform the first $n-k$ projection stages using ${\tt ProjOp}$ and starting with $A$.  

Set $\bm{P_1}$ to be the projection polynomials with main variable $x_i$

\If{$k\neq1$}{
$\mathcal{D}_{k-1} \leftarrow {\tt CADAlgo}({\bf P_{k-1}}, [x_1,\ldots,x_{k-1}])$
\tcp*{Computation of a CAD of $\mathbb{R}^{k-1}$}
}

\If{ $\mathcal{D}_{k-1}$ = {\rm FAIL}}{
\Return {\rm FAIL}  \tcp*{${\bf P_{k-1}}$ is not well oriented}
}

$\mathcal{D}_k \leftarrow []$
\label{step:RL1}\;

\eIf{$k=1$}{
Set $S$ to be the CAD formed by decomposing $\mathbb{R}$ according to the roots of $\bm{P_1}$.\;
  \If{$|S| > 1$}{
    \For{$i=1 \dots (|S|-1)/2$}{
      $\mathcal{D}_k.{\tt append}(S[2i])$ \tcp*{Cells with even index included}
    }
}
}{
\For{$c \in \mathcal{D}_{k-1}$}{
  $S \leftarrow {\tt GenerateStack}(\bm{P_k}, c)$\tcp*{$k$th lifting stage}
  \If{ $S$ = {\rm FAIL}}{
  \Return {\rm FAIL}  \tcp*{${\bf P_{k}}$ is not well oriented}
  }
  \If{$|S| > 1$}{
    \For{$i=1 \dots (|S|-1)/2$}{
      $\mathcal{D}_k.{\tt append}(S[2i])$ \tcp*{Cells with even index included}
      \label{step:RL2}
    }
  }
}
}
\For{$i = k+1, \dots n$ \label{step:FL1}}{
$\mathcal{D}_i \leftarrow []$\;
\For{$c \in \mathcal{D}_{i-1}$}{
  $S \leftarrow {\tt GenerateStack}(\bm{P_i}, c)$\tcp*{$i$th lifting stage}
  \If{ $S$ = {\rm FAIL}}{
    \Return {\rm FAIL}  \tcp*{${\bf P_{i}}$ is not well oriented}
  }
  $\mathcal{D}_i.{\tt append}(S)$ \tcp*{All cells included}
  \label{step:FL2}
  }
}
\Return $\mathcal{D}_n$\;

\caption{
Algorithm to produce {\scad}s with respect to a variety of lower dimension.
}
\label{alg:V2}
\end{algorithm}

Example \ref{ex:MCAD2} demonstrates the savings offered by a variety of lower dimension, but also the difficulty in obtaining an actual V-{\scad}.

\begin{example}{\ }\\
\label{ex:MCAD2}
Consider again the polynomials from Example \ref{ex:MCAD} (with the same variable ordering) but this time with the formula $\varphi_2 := f<0 \land g=0$.  There is only one cell where this is true (the $y$-axis inside the circle).  A minimal V-{\scad} truth-invariant for $\varphi_2$ would contain only the 5 cells shown on the left of Figure \ref{fig:MCAD2}.  However, if we use Algorithm \ref{alg:V2} then we would produce 11 cells, as shown on the right.  
Algorithm \ref{alg:V2} first builds a CAD of $\mathbb{R}^1$ with respect to all the univariate projection polynomials.  This has 7 cells (the points $-1,0,1$ and the intervals in-between).  It discards the intervals and lifts over the three points with respect to $f$ obtaining the 11 cells shown.  No lifting (and hence real root isolation) was performed over the other 4 cells in $\mathbb{R}^1$ since we could conclude they were not part of the variety.
\end{example}

\begin{figure}[ht]
\begin{center}
\includegraphics[width=0.45\textwidth]{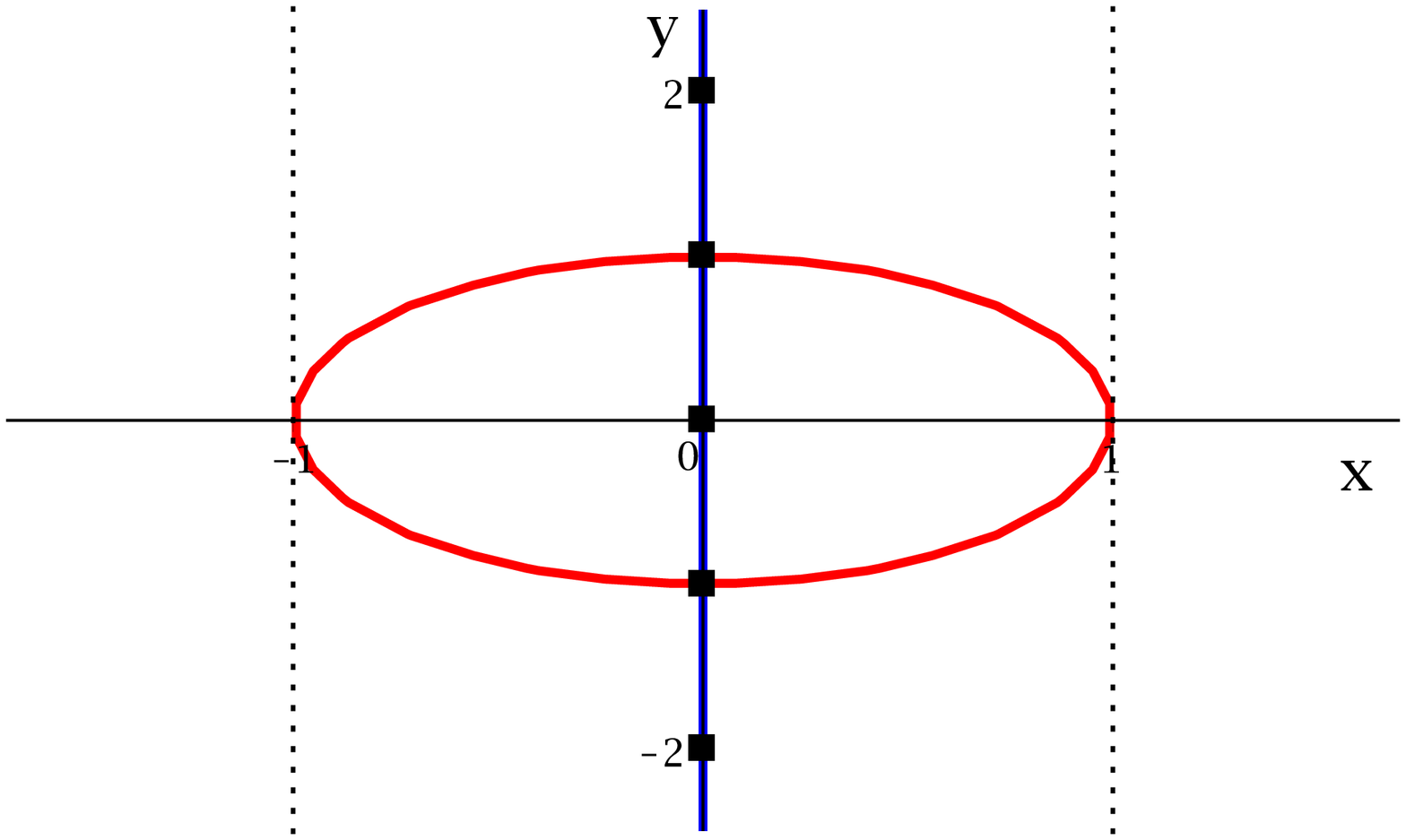}
\includegraphics[width=0.45\textwidth]{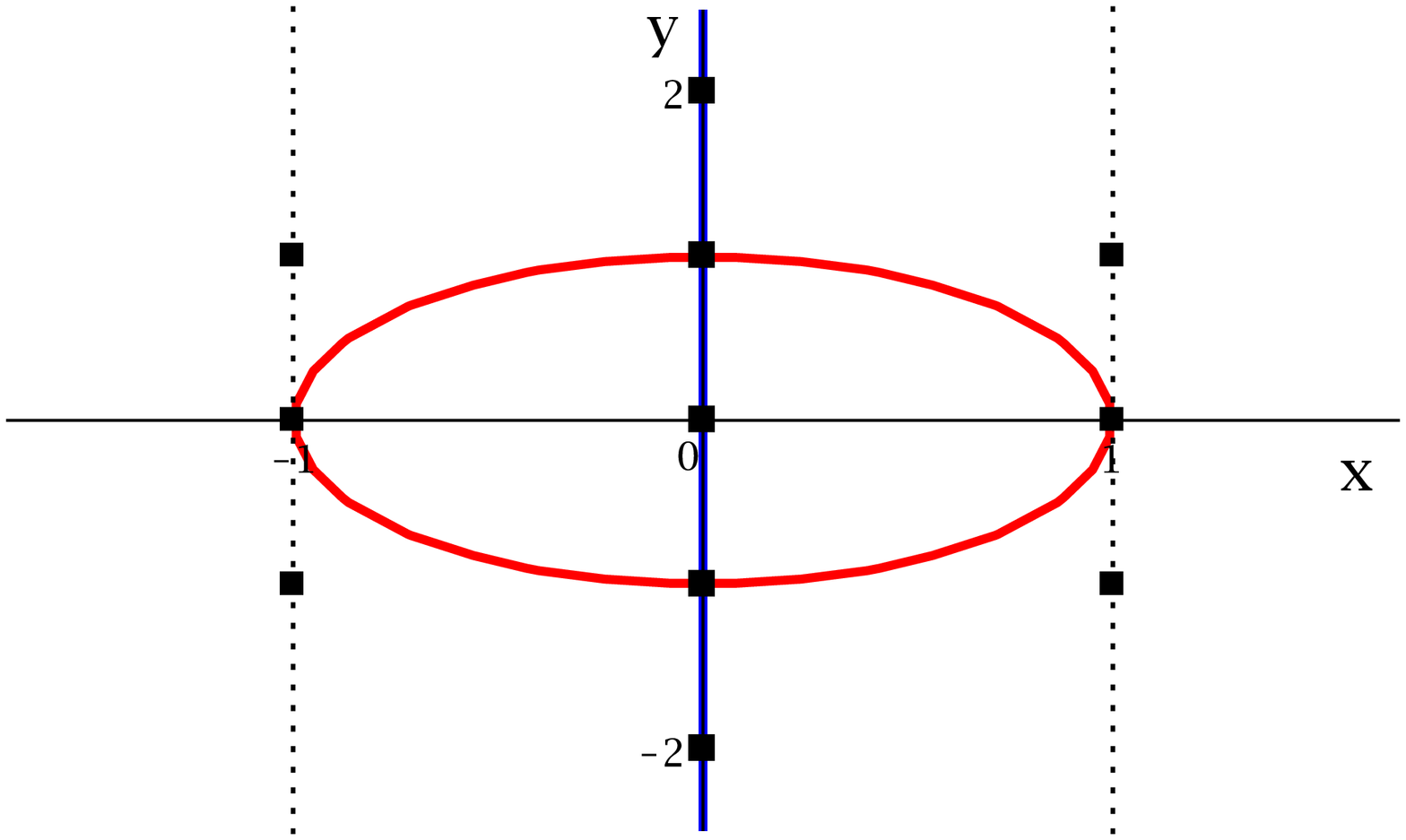}
\end{center}
\caption{Images representing the CADs described in Example \ref{ex:MCAD2}.  The minimal variety sub-CAD is on the left with the output from Algorithm \ref{alg:V2} on the right. 
%The solid boxes represent the cells constructed in $\mathbb{R}^2$.
} 
\label{fig:MCAD2}
\end{figure}

To allow factors of the equational constraint with different main variable would require a further extension to perform multiple stages of restricted lifting (steps $\ref{step:RL1}-\ref{step:RL2}$ of Algorithm \ref{alg:V2}) when lifting to $\mathbb{R}^i$ where $x_i$ is a main variable of a factor, and full lifting (steps $\ref{step:FL1}-\ref{step:FL2}$) otherwise.

Finally, note that we have only discussed using a single (designated) equational constraint.  
% and we are therefore working on a variety of (real) co-dimension 1.  
In the case of two or more (as in $f_1=0\land f_2=0\land\ldots$) we need the theory of bi-equational constraints and beyond.  Although there has been some work on this \cite{McCallum2001, BM05} we have not investigated its interaction with variety {\scads} yet.  

%There are various options for utilising the presence of multiple equational constraints in different clauses of the formula, the simplest being to multiply them together to form a single overall constraint.  Another option is to split the formula up into sub-formulae each with a single equational constraint and use a variety TTICAD \cite{BDEMW13} as discussed in Section \ref{sec:interaction}.  In both cases, a problem's tractability can differ with the choice of equational constraint, and heuristics for choosing the most efficient equational constraint were discussed in \cite{BDEW13}, but have not yet been investigated with variety CAD.  

%\TODO{Changed wording of this last paragraph - are we happy?}

\subsection{Layered {\scads}}
\label{sec:LCAD}

The idea of returning CAD cells of certain dimensions was first discussed in \cite{McCallum1993} and revisited in \cite{Strzebonski00, Brown2013}. All these papers discuss the idea of returning only CAD cells of full-dimension, noting that this is sufficient to solve problems involving only strict polynomial inequalities.  
This idea was extended in the technical report  \cite{WE13} to the case of returning cells with dimension above a prescribed value. We reproduce some of that unpublished work here, including the key algorithms. 
%We concentrate on algorithms which produce sign-invariant CADs but note that the ideas can easily be applied to other invariance properties.
% using Collins' original projection algorithm \cite{Collins:1975ur} along with McCallum's improved projection operator \cite{McCallum:1988ir}. 

\begin{definition}{\ }\\
\label{def:Lscad}
Define the set of cells in a CAD of a given dimension as a {\bf layer}.    
Let $\ell$ be an integer with $1 \leq \ell \leq n+1$. Then an {\bf $\bm{\ell}$-layered Sub-CAD ($\bm{\ell}$-L-sub-CAD)} is the subset of a CAD of dimension $n$ consisting of all cells of dimension $n-i$ for $0 \leq i < \ell$.
We refer to a CAD consisting of all cells of all dimensions as a {\bf complete CAD}.
\end{definition}

\begin{remark}{\ }
\begin{enumerate}
\item An $\ell$-layered CAD consists of the top $\ell$ layers of cells in a CAD.  The dimensions of these cells will depend on the dimension of the space the CAD decomposes.

\item In the literature the set of cells of full-dimensional has been referred to as an open CAD, a full CAD and a generic CAD.  We prefer layered CAD as it is less open to misinterpretation and allows us to generalise the idea beyond the top dimension.  The set of cells of full dimension is then a $1$-layered CAD.
Note also that a complete CAD of $\mathbb{R}^n$ has $(n+1)$-layers.

%An $\ell$-layered CAD returns the $\ell$ layers of maximal dimension, whilst a variety CAD will in the process of restricting to a variety, discarding the full-dimensional layer of cells, and returning a subset of the remaining layers.

\item When building a $1$-layered CAD it was pointed out in \cite{Strzebonski00} that a simplified projection operator could be used.  Instead of taking the full set of coefficients for a polynomial only the leading coefficient is required (since the others are there to ensure delineability if the first vanishes, but this could only happen on a cell of less than full dimension).  In this paper we focus on improvements to the lifting phase, but if only a $1$-layered CAD is required then this further saving in the projection phase is available.
\end{enumerate}
\end{remark}

Algorithm \ref{alg:layeredCAD} describes how an $\ell$-layered CAD may be produced.  The main idea is that during the lifting process we check cell dimension before each stack construction.  If a cell has too low a dimension to give cells in $\mathbb{R}^n$ of dimension $\ell$ or higher then it is discarded.  We give a general algorithm which can be used with any suitable (and compatible) projection operator and stack generation procedure.  

\begin{algorithm}
%\DontPrintSemicolon
\SetKwInOut{Input}{Input}\SetKwInOut{Output}{Output}

\Input{A formula $\varphi$, an integer $1 \leq \ell \leq n+1$ and variables $\bm{x}=x_1, \dots x_n$.  $\varphi$ is in $\bm{x}$.
}

\Output{An $\ell$-layered {\scad} for $\varphi$, or FAIL.}
\BlankLine

%$\mathcal{D} \leftarrow [ \ ]$\;

${\bf P} \leftarrow$ output from applying ${\tt ProjOp}$ repeatedly to $\varphi$
\tcp*{Full projection phase}
\For{$i=1,\ldots,n$}{
  Set ${\bf P}[i]$ to be the projection polynomials with mvar $x_i$\;
  %${\bf P}[i+1] \leftarrow {\tt ProjOp}({\bf P}[i])$ 
  %\tcp*{Split according to main variable}
}

Set $\mathcal{D}[1]$ to be the CAD of $\mathbb{R}^1$ obtained by isolating the roots of  ${\bf P}[1]$\;

\For{$i=2,\ldots,n$}{
  $\mathcal{D}[i] \leftarrow [ \ ]$\;
  \For{$c \in \mathcal{D}[i-1]$}{
    ${\tt dim} \leftarrow \sum_{\alpha \in c.{\tt index}} \left(\alpha \mod 2\right)$
    \tcp*{Lift over suitable dimension cells}
    \If{${\tt dim} > i-\ell-1 $ \label{step:discard1}}{
    	$S \leftarrow {\tt GenerateStack}({\bf P}[i],c)$\;
    	\eIf{$S=${\rm FAIL}}
    		{\Return {\rm FAIL}  \tcp*{Input is not well oriented}}{ $\mathcal{D}[i].{\tt append}(S)$}
     }
  }
}

$\mathcal{D} \leftarrow [ \ ]$\;

\For{$c \in \mathcal{D}[n]$}{
    ${\tt dim} \leftarrow \sum_{\alpha \in c.{\tt index}} \left(\alpha \mod 2\right)$\;
    \If{${\tt dim} > n-\ell$ \label{step:discrad2}}{
       $\mathcal{D}.{\tt append}(c)$\tcp*{Remove cells of low dimension from final lift}
     }
}
\Return $\mathcal{D}$\;

\caption{${\tt LayeredSubCAD}(\varphi,\ell,{\bf x})$: Algorithm to produce $\ell$-layered {\scad}s.
%: {\tt CADNLayered(F,l,vars)}
}
\label{alg:layeredCAD}
\end{algorithm}

\begin{theorem}
\label{thm:L1}
Algorithm \ref{alg:layeredCAD} satisfies its specification.
\end{theorem}
\begin{proof}
The output being a subset of a valid CAD follows from the correctness and compatibility of the sub-algorithms used (as proven in \cite[etc.]{Collins1975, Hong1990, McCallum1998, McCallum1999}).  It remains to verify that the cells discarded could not contribute cells in the top $\ell$ layers of the outputted {\scad}.

When lifting over a cell of dimension $d$ in a {\scad} of $\mathbb{R}^i$ it can contribute cells in the {\scad} of $\mathbb{R}^n$ of dimension at most $d + n - i$.  If these are required for an $\ell$-layered {\scad} they must have dimension at least $n-\ell$, and so we can discard them if $d\leq i-\ell-1$ (step \ref{step:discard1}).  

When performing the final lift we must build stacks over cells of dimension $n-\ell-1$ or greater, but of course some of the cells in those stacks may not have dimension $n-\ell$.  Hence we check for this at the end (step \ref{step:discrad2}) only keeping those of the required dimension.
\end{proof}

\begin{example}{\ }
\label{ex:LCAD} \\  
Consider once again the polynomials introduced by Example \ref{ex:MCAD} (with variable ordering $x \prec y$). The first image in Figure \ref{fig:MCAD} demonstrated that the simplest sign-invariant complete CAD for the polynomials would have 23 cells.  Figure \ref{fig:LCAD} shows the same CAD, this time with the dimensions of each cell indicated. 

\begin{figure}[ht]
\begin{center}
\includegraphics[width=0.45\textwidth]{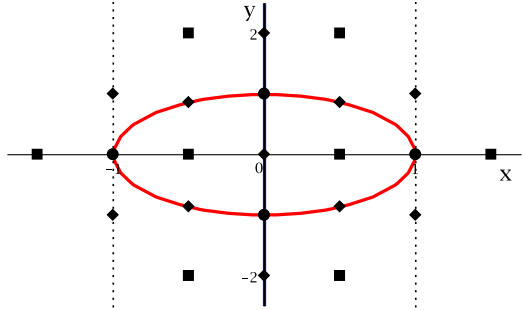}
\end{center}
\caption{Image representing the CAD described in Example \ref{ex:LCAD}.  The boxes indicate cells of dimension 2, the diamonds dimension 1 and the spheres dimension 0.}
\label{fig:LCAD}
\end{figure}

Suppose that the polynomials arise from a problem involving the formula $\varphi_3 := f<0 \land g<0$.  The solution is then given by the single cell inside the circle to the left of the $y$-axis. In this case we know the solutions must all be cells of full-dimension and thus that a 1-layered CAD would suffice.  Algorithm \ref{alg:layeredCAD} would return only the 8 cells of dimension 2 (the boxes in Figure \ref{fig:LCAD}).
\end{example}

In Example \ref{ex:LCAD} a 1-layered {\scad} is $\varphi$-sufficient.  This was a case where the problem consisted of strict inequalities meaning the solutions have full-dimension (as noted previously in \cite{McCallum1993, Strzebonski00}).  More generally there are classes of problems with known solution dimension for which the layered sub-CAD technology will be beneficial.  For example, recall the cyclic $n$-polynomials (the set of $n$ symmetric polynomials in $n$ variables).  In \cite{Backelin1989} it is shown that if there exists $m>0$ such that $m^2|n$ then there are an infinite number of roots and they are, at least, of dimension $(m-1)$.  A sub-CAD containing cells of dimension $(m-1)$ and higher would therefore be sufficient to identify the largest families of solutions, and so a layered sub-CAD could be an appropriate tool.

Layered {\scad}s will also be useful when, although a {\scad} may not be $\varphi$-sufficient, the underlying application only requires generic families of solutions.  For example, when robot motion planning we need only identify paths through cells of full-dimension, and so a 1-layered sub-CAD may be enough.  We consider a motion planning example later in Section \ref{SUBSEC:Piano}.  In that example we actually need the full-dimensional cells on a variety, which are part of the 2-layered sub-CAD.  We find that the 1-layered variety sub-CADs developed in Section \ref{subsec:LMCAD} are also appropriate. 
%However, to analyse the adjacency of such paths we may need the cells of dimension one lower and thus a 2-layered CAD may be the appropriate choice.  

We now explain how Algorithm \ref{alg:layeredCAD} may be adapted to a recursive procedure.  
Our general principle when constructing layered {\scads} is to stop lifting over a cell if it cannot lead to cells of sufficient dimension in $\mathbb{R}^n$.  This occurs when the cell in question was produced as a section, rather than a sector.  Let us call these {\bf terminating sections}.
Instead of discarding these sections after they are produced, the recursive algorithm stores them in a separate output variable for use later if required.

For example, when constructing a 1-layered CAD these terminating sections will have dimension $m-1$ where $m$ is the number of variables lifted at the level the section is constructed (i.e. they are deficient by one dimension). So at the final lifting stage the terminating sections will have dimension $n-1$.  Now suppose we wish to extend to a $2$-layered CAD. If we take one of these terminating sections of any dimension and construct successive stacks over it we will obtain cells that are of dimension $n-1$ and terminating sections that are deficient (with respect to their level) by two dimensions. Combining these $n-1$-dimensional cells  with the 1-layered CAD produces a 2-layered CAD.  If we again store the remaining (and new) terminating sections then they can be used later to construct a 3-layered CAD if desired. In this manner we can recursively produce $\ell$-layered CADs, useful if the numbers of layers required is not known at the start of the computation.  The method is described by Algorithm \ref{alg:recursivelayered}. 

\begin{algorithm}
%\DontPrintSemicolon
\SetKwInOut{Input}{Input}\SetKwInOut{Output}{Output}

\Input{A formula $\varphi$, variables $\bm{x} = x_1, \dots, x_n$, a list of cells of $\mathbb{R}^n$ labelled $\mathcal{LD}$, and a list of lists of cells, $\mathcal{C}$ (ordered by increasing dimension).  The formula $\varphi$ is in $\bm{x}$.  
The lists $\mathcal{C}$ and $\mathcal{LD}$ may be empty.  Otherwise the list $\mathcal{LD}$ must contain a layered {\scad} for $\varphi$ and $\mathcal{C}$ the corresponding terminating sections.}

\Output{Either FAIL or a layered {\scad} $\mathcal{D}'$ for $\varphi$.  If $\mathcal{LD}$ was empty then $\mathcal{D}$ is a 1-layered {\scad} and otherwise it is a layered {\scad} with including cells of dimension one layer lower than $\mathcal{LD}$.  Also, an updated list of lists of cells containing the terminating sections sufficient to construct the complete CAD.}
\BlankLine

{\tt global} {\bf P} \tcp*{To avoid recomputing projection polynomials}

\If{{\bf P} is undefined}{
${\bf P} \leftarrow$ output from applying ${\tt ProjOp}$ repeatedly to $\varphi$
\tcp*{Full projection phase}
  \For{$i=1,\ldots,n$}{
    Set ${\bf P}[i]$ to be the projection polynomials with mvar $x_i$\;
  }
%${\bf P} \leftarrow [ F ]$\;
%\For{$i=1,\ldots,n-1$}{
%  ${\bf P}[i+1] \leftarrow {\bf Proj}({\bf P}[i])$\;
%}
}

$\mathcal{C}' \leftarrow [ \ ]$\;

\eIf{$\mathcal{C}==\emptyset$}{
$\mathcal{D} \leftarrow [ \ ]$ \tcp*{Base case - construct $\mathbb{R}^1$}
Set ${\tt Base}$ to be the CAD of $\mathbb{R}^1$ obtained by isolating the roots of  ${\bf P}[1]$\;
%ME: IT WAS "({\bf P}[n])" - TYPO I THINK
\For{$i=1,\dots,{\tt length}({\tt Base})$}{
\eIf{$(i \mod 2) == 1$}{
  $\mathcal{D}[1].{\tt append}({\tt Base}[i])$
  \tcp*{Sectors only (odd index) for the 1-layered CAD}
}{
  $\mathcal{C}'[1].{\tt append}({\tt Base}[i])$
  \tcp*{Sections (even index) added stored}
}
}
$\mathcal{D'} \leftarrow [ \ ]$\;
}{
$\mathcal{D} \leftarrow \mathcal{C}$ \tcp*{Use previously computed terminating sections}
$\mathcal{D'} \leftarrow \mathcal{C}[n]$\;
}

\For{$i=2,\ldots,n$}{
  \For{$c \in \mathcal{D}[i-1]$}{

       $S \leftarrow {\tt GenerateStack}({\bf P}[i],c)$\;
  	   \If{ $S$ = {\rm FAIL}}{
         \Return {\rm FAIL}  \tcp*{Input is not well oriented}
        }       
       \For{$j=1,\ldots,{\tt length}( S )$}{
       \eIf{$(j \mod 2) == 1$}{
       $\mathcal{D}[i].{\tt append}(S[j])$ 
       \tcp*{Sector (odd index) so add to output CAD}
       }{
       $\mathcal{C}'[i].{\tt append}(S[j])$ 
       \tcp*{Section (even index) so store}
       }

       }

  }
}

$\mathcal{D}' \leftarrow \mathcal{D}' \cup \mathcal{LD}$ \tcp*{Combine new cells with those previously computed}

\Return $[\mathcal{D}'$, $\mathcal{C'}]$\;

\caption{${\tt LayeredSubCADRecursive}(\varphi, {\bf x}, \mathcal{C}, \mathcal{LD})$: Algorithm that may be applied recursively to produce layered {\scads} with increasing numbers of layers. 
%based on Algorithm \ref{alg:layeredCAD}: 
%${\tt {CADRecursiveLayered}(F,vars,\mathcal{C},\mathcal{LD})}$
}
\label{alg:recursivelayered}
\end{algorithm}

\begin{theorem}
\label{thm:L2}
Algorithm \ref{alg:recursivelayered} satisfies its specification.
\end{theorem}
\begin{proof}
As with the proof of Theorem \ref{thm:L2} the correctness mostly follows from the correctness of the sub-algorithms used and it remains only to verify that the correct layers are produced.  We see that if $\mathcal{LD}$ is empty then a CAD of the real line is produced and the sections and sectors are separated.  Otherwise, each of the cells in the existing terminating sections are lifted over.  The stacks constructed consist of sectors, which together with $\mathcal{LD}$ form a layered {\scad} with one extra layer, or new terminating sections which could be lifted over to gain the complete CAD.
\end{proof}

Algorithm \ref{alg:recursivelayered} has been implemented in {\sc Maple}. It uses a global variable to avoid recalculation of projection polynomials, and {\sc Maple}'s unevaluated function call syntax (prefixing the command with ${\tt \%}$) allowing it to be repeatedly evaluated to return {\scad}s with an increasing number of layers.

\newpage

We finish this subsection by noting the following interesting property of certain layered CADs.
\begin{theorem}[\cite{WE13}]
\label{thm:twolayeredorderinvariant}{\ }\\
Let $F \subset \mathbb{Z}[x_1,\ldots,x_n]$ and $\mathcal{D}$ be a 1- or 2-layered CAD of $\mathbb{R}^n$ sign-invariant for $F$.
%Let $F \subset \mathbb{Z}[x_1,\ldots,x_n]$ and let $\mathcal{D}$ be a 1-layered or 2-layered CAD of $\mathbb{R}^n$ sign-invariant with respect to $F$. 
Then $\mathcal{D}$ is order-invariant with respect to $F$, meaning each polynomial has constant order of vanishing on each cell.
\end{theorem}
Order-invariance is a stronger property than sign-invariance, but the extra knowledge it gives allows for the validated use of smaller projection operators (see for example \cite{McCallum1998}).  Hence this property allows for the avoidance of well-orientedness checks during stack generation when building 1 or 2-layered {\scads}.  This means not just a savings in computation time but the avoidance of unnecessary failure declarations that can sometimes follow from such checks (as discussed in \cite{Brown2005a, England13a}).

\subsection{Other sub-CADs in the literature}
\label{subsec:furtherscads}

We note that other ideas from the CAD literature could be described as, or be easily adapted to produce, {\scad}s.  We list these for completeness.

In \cite{McCallum1997}, whilst trying to solve a motion planning problem in the plane described by a formula $\varphi$, the author identifies a subset of cells in the decomposition of $\mathbb{R}^1$ for which any valid cell for $\varphi$ must lie over.  Lifting over only these cells only gives a $\varphi$-sufficient {\scad}.  Similar ideas are in \cite{IYAY09}.

In \cite{Brown2013} an algorithm is presented which given polynomials $F$ and a point $\alpha$, returns a single cell containing $\alpha$ on which $F$ is sign-invariant.  The cell belongs to a CAD (although not necessarily one that could be produced by any known algorithm) and hence this is an extreme example of a {\scad}. 

Partial CAD \cite{CH91} works by avoiding the splitting of cylinders into stacks when lifting over cells where the truth value is already known.  If cells over which the truth value is false were instead simply discarded then what is left would be a {\scad} sufficient to analyse the input formula.  

In \cite{SS03} an algorithm is described which takes polynomials F and returns a CAD $D$ and theory $\Theta$ (set of negated equations).  The CAD is sign-invariant for $F$ for all points which satisfy $\Theta$.  Rather than a {\scad} of $\mathbb{R}^n$ this is actually a CAD of $\mathbb{R}^n_{\Theta}$ % = \{ \bm{x} \in \mathbb{R} \mid x \in \Theta \}$, that is 
: all those points in $\mathbb{R}^n$ except the set of measure zero which do not satisfy $\Theta$.

In \cite{Strzebonski12}, an algorithm is given for solving systems over cylindrical cells described by cylindrical algebraic formulae. This allows cells produced from a CAD to be used easily in further computation, which may implicitly be used to produce either {\scads} or CADs of a sub-space.

\section{Extending the use of {\scad}}
\label{sec:interaction}

Layered and variety {\scads} can offer significant savings individually but this can be increased by combining them with each other, as discussed in Section \ref{subsec:LMCAD}.  We then consider how they may interact with the recent theory of truth-table invariant CAD in Section \ref{subsec:LMTTICAD}.

\subsection{Combining layered and variety sub-CADs}
\label{subsec:LMCAD}

The idea behind {\scads} is simple: we wish to filter out only those cells of relevance to us. A variety {\scad} does this in one step during a single stage of the lifting phase, whereas a layered {\scad} stratifies the cells throughout the whole lifting process. 
%A useful analogy is that variety lifting separates cells in a vertical manner, whereas layered lifting sorts in a horizontal manner, filtering out cells at every step.
There is no reason why these two ideas cannot be combined.  We discuss this in the case where all factors of the equational constraint defining the variety have main variable $x_n$ and are not nullified on a low-dimensional cell.  (It may be generalised but this would require caution as discussed in Section \ref{sec:MCAD}.)  
The key is the following simple result.
\begin{lemma}{\ }
\label{lem:dimensiononvariety} \\
Let $\mathcal{D}$ be a CAD of $\mathbb{R}^{n-1}$ for some formula $\varphi$ and let $c$ be a cell in $\mathcal{D}$ of dimension $k$. Further, suppose $f=0$ is an equational constraint of $\varphi$, each factor of $f$ has main variable $x_n$, and $f$ is not nullified on $c$. % no vertical components
Then any section of the stack lifted over $c$ with respect to $f$ will have dimension $k$.
\end{lemma}
The lemma shows that lifting over a cell onto a variety gives cells with the same dimension.  Hence when lifting over an $\ell$-layered sub-CAD of $\mathbb{R}^{n-1}$ (which contains cells of dimension $n-\ell, \ldots, n-1$) we produce a {\scad} containing all the cells of dimensions $n-\ell, \ldots, n-1$ on the variety.  We may think of this as an $\ell$-layered {\scad} {\em of the variety}. %We formalise this in a definition:
\begin{definition}{\ }\\
\label{def:LMscad}
Let $\varphi$ be a Tarski formula with equational constraint $f=0$ which has main variable $x_n$ in all its factors, and let $1 \leq \ell \leq n$. A truth-invariant {\scad} for $\varphi$ whose cells have dimension $n-i-1$ for $0 \leq i < \ell$ and rest on the variety defined by $f=0$ is an {\bf {\lLayManScad{$\bm{\ell$}}} ({\lLMscad{$\bm{\ell}$}})}. 
\end{definition}

\begin{remark}{\ }\\
Note that in general an {\llaymanscad{$\ell$}} consists of the top $\ell$ layers of cells {\em on the variety}. This can be thought of as the intersection of an $(\ell+1)$-layered CAD of $\mathbb{R}^n$ with the variety (as the layer of $n$-dimensional cells is discarded when lifting to the variety).
\end{remark}
\noindent Lemma \ref{lem:dimensiononvariety} leads to Algorithm \ref{alg:LVSubCAD} for producing {\lLMscad{$\bm{\ell}$}}s.  

%We use Algorithm \ref{alg:layeredCAD} to produce an $\ell$-layered {\scad} of $\R^{n-1}$ and then lift to give a variety {\scad} in the same way as Algorithm \ref{alg:VarietySub-CAD}.

\begin{algorithm}[hb]
%\DontPrintSemicolon

\SetKwInOut{Input}{Input}\SetKwInOut{Output}{Output}

\Input{A formula $\varphi$, a declared equational constraint $f=0$ from $\varphi$, an integer $1 \leq \ell \leq n+1$ and variables $\bm{x} = x_1, \dots, x_n$.  $\varphi$ is in $\bm{x}$ and all factors of $f$ have main variable $x_n$.
}
\Output{An $\ell$-layered variety {\scad} $\mathcal{D}$ for $\varphi$, or FAIL.}
\BlankLine

Extract from $\varphi$ the set of polynomials $A$ and from $f$ the subset $E \subset A$ \;

${\bf P} \leftarrow$ output from applying ${\tt ProjOp}$ to (A,E)
\tcp*{First projection stage}

$\mathcal{D}' \leftarrow {\tt LayeredSubCAD}({\bf P}, \ell)$
\tcp*{Computation of a sub-CAD of $\mathbb{R}^{n-1}$}

\If{ $\mathcal{D}'$ = {\rm FAIL}}{
\Return {\rm FAIL}  \tcp*{${\bf P}$ is not well oriented}
}

$\mathcal{D} \leftarrow []$\;

\For{$c \in \mathcal{D}'$}{
  $S \leftarrow {\tt GenerateStack}( E, c)$\tcp*{Final lifting stage}
  \If{ $S$ = {\rm FAIL}}{
  \Return {\rm FAIL}  \tcp*{Input is not well oriented}
  }
  \If{$|S| > 1$}{
    \For{$i=1 \dots (|S|-1)/2$}{
      $\mathcal{D}.{\tt append}(S[2i])$ \tcp*{Cells with even index are sections}
    }
  }
}

\Return $\mathcal{D}$\;

\caption{\texttt{LayeredVarietySubCAD}$(\varphi, f, \ell, \bm{x})$: Algorithm for $\ell$-layered variety {\scads}.}
\label{alg:LVSubCAD}
\end{algorithm}

\begin{theorem}
\label{thm:LV}
When the sub-algorithms are chosen to implement McCallum's algorithm to produce CADs with respect to an equational constraint \cite{McCallum1999}, then Algorithm \ref{alg:LVSubCAD} satisfies its specification with the outputted {\scad} consisting of cells on which the input formula has constant truth value.
\end{theorem}

\begin{proof}
As with Theorem \ref{thm:MCAD} the {\scad} structure and invariance property follow from the results in \cite{McCallum1999}.  Theorem \ref{thm:MCAD} verifies that $\mathcal{D}$ is an $\ell$-layered {\scad} of $\mathbb{R}^{n-1}$ and Lemma \ref{lem:dimensiononvariety} concludes that the lifting in the final loop results in an $\ell$-layered {\scad} of $\mathbb{R}^{n}$.  The case where Lemma \ref{lem:dimensiononvariety} does not hold is a case where the input is not-well oriented and thus FAIL is returned.
\end{proof}

We now give a simple example which uses both the layered and variety sub-CAD ideas together, as well as illustrating the difference variable ordering can make.
\begin{example}{\ }
\label{ex:MLCAD} \\
Consider a final time the polynomials from Example \ref{ex:MCAD}
%\[
%f := x^2 + y^2 - 1, \qquad g := x
%\]
and this time the formula $\varphi_1:= f=0 \land g<0$.  In Example \ref{ex:MCAD} we saw that with variable ordering $x \prec y$ a variety {\scad} could be produced with 8 cells, 3 of which described the solutions.  Instead let us build a {\llaymanscad{1}} as represented in the first image of Figure \ref{fig:MLCAD}.  The output would now be only 4 cells (those indicated with solid boxes).  Two of these describe the generic solution sets $\{x \in (-1,0), y = \pm \sqrt{1-x^2} \}$ but the third and final cell in the solution set $\{x=-1,y=0\}$ has been lost.  Note that in this case (unlike the variety {\scad}) there will be a reduction in CAD computation time as there will be no lifting over cells of dimension zero in the CAD of $\mathbb{R}^1$.  The cells which have been computed but discarded are shown by empty boxes.

\begin{figure}[ht]
\begin{center}
\includegraphics[width=0.45\textwidth]{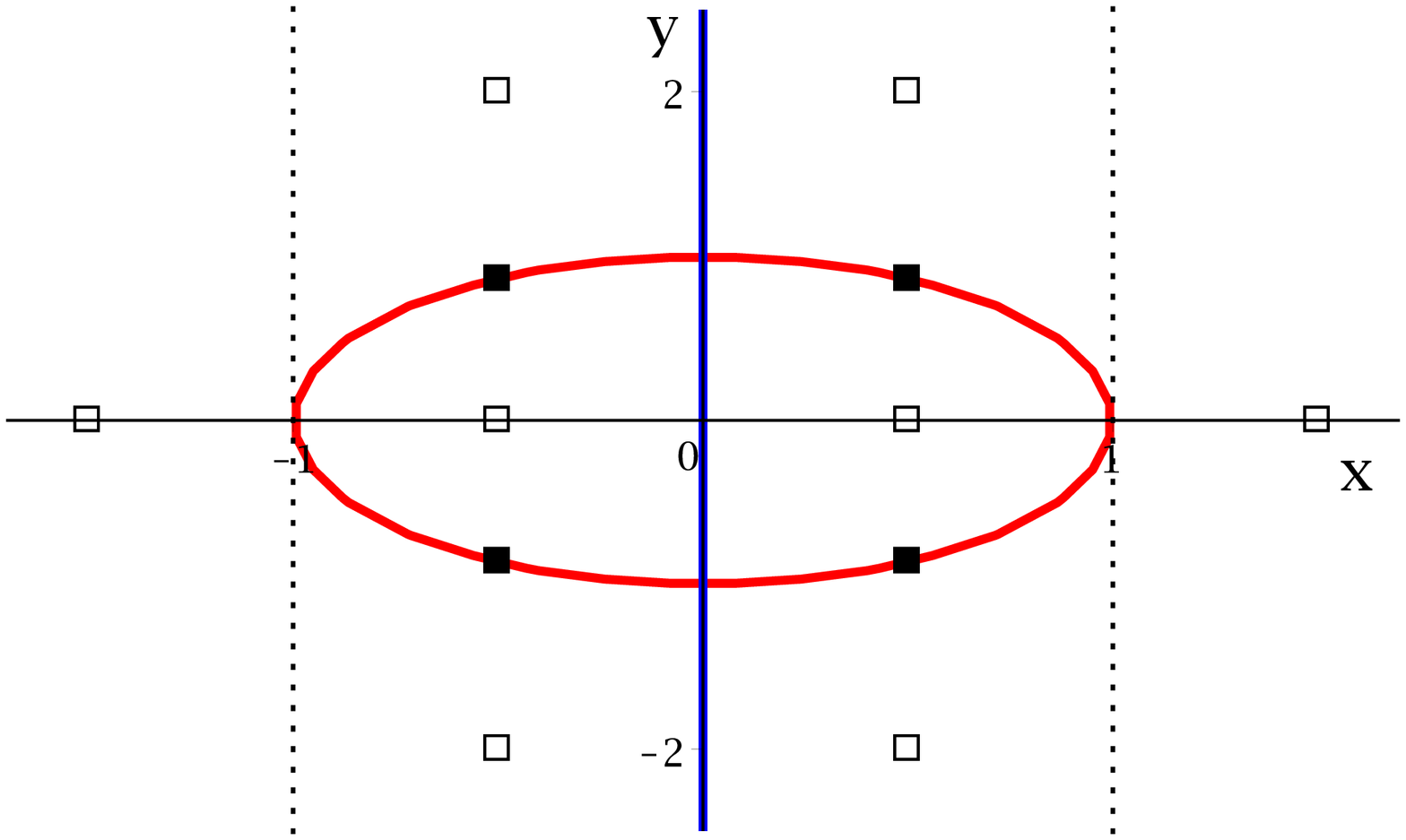} \\
\includegraphics[width=0.45\textwidth]{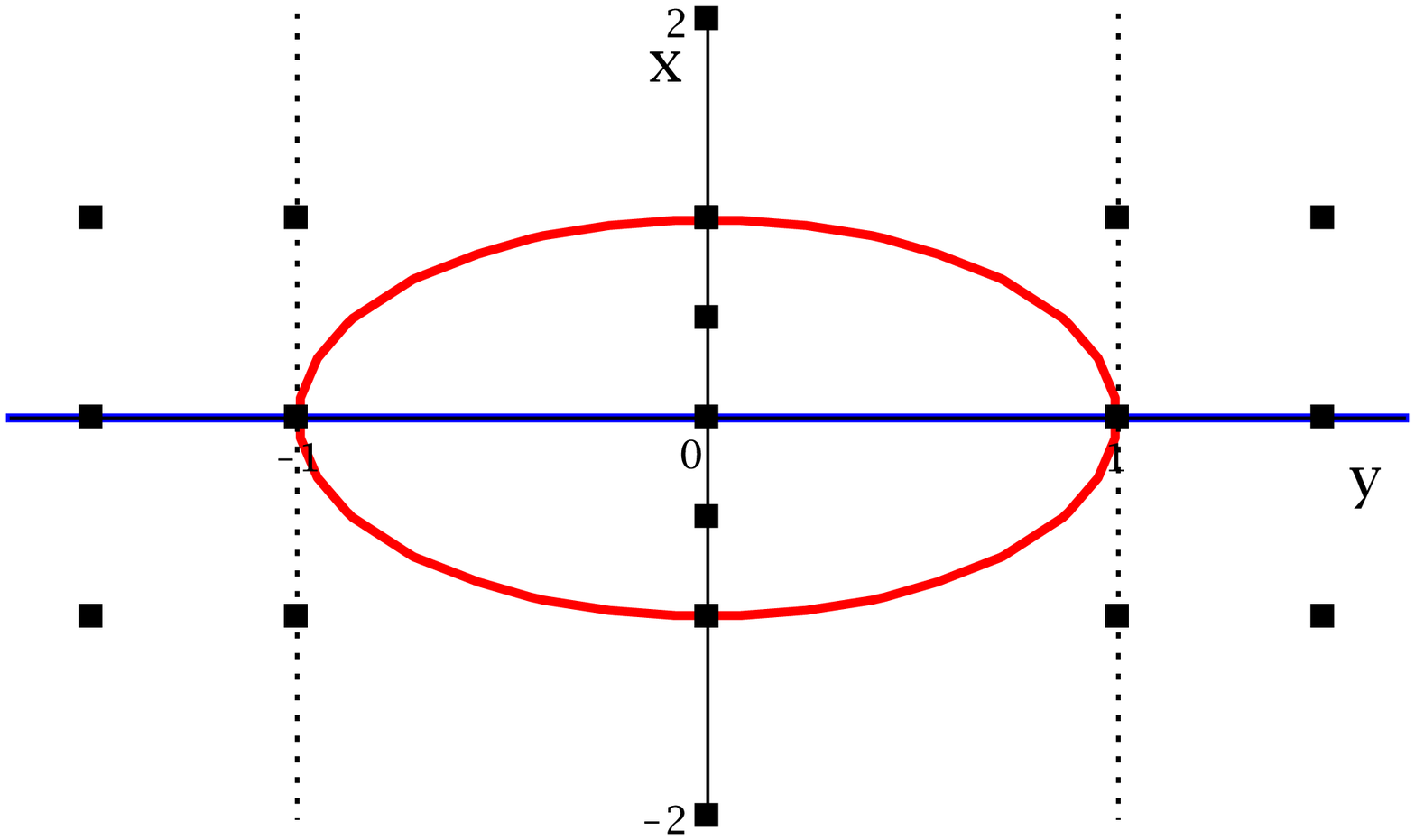}
\includegraphics[width=0.45\textwidth]{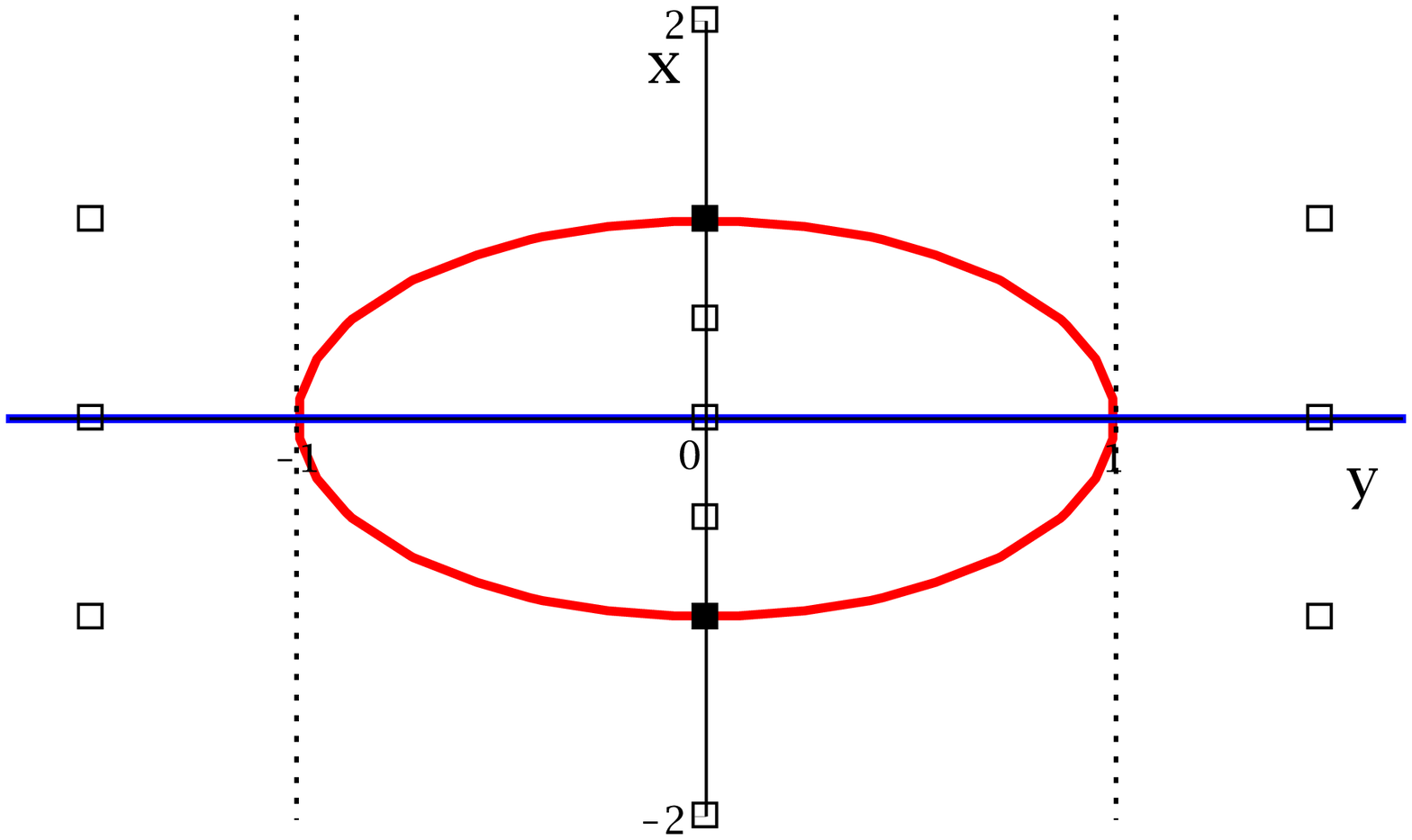}
\end{center}
\caption{Figure representing the CADs described in Example \ref{ex:MLCAD}.  The solid boxes represent cells returned and the empty boxes cells computed but discarded.  The first image uses variable ordering $x \prec y$ and the others $y \prec x$ (axes labelled accordingly).}
\label{fig:MLCAD}
\end{figure}

Consider now the alternative variable ordering $y \prec x$.  In this case a minimal complete sign-invariant CAD must have 19 cells, as represented by the second image in Figure \ref{fig:MLCAD}.  A variety {\scad} would return the 4 cells describing the circle and a {\llaymanscad{1}} only the 2 cells of these which have full-dimension on the variety.  This {\scad} is represented by the final image in Figure \ref{fig:MLCAD}.  Again, there has been a saving in CAD construction time but in this case the cells returned include the full solution set (the single cell describing the upper half of the circle).
\end{example}

%The complexity of  a {\llaymanscad{1}} is discussed in Section \ref{sec:Complexity} and an example of a problem which can be partly tackled by {\LMscad} is given later in Section \ref{SUBSEC:LMCASD}.

\subsection{Truth table invariant sub-CADs}
\label{subsec:LMTTICAD}

The new ideas of variety and layered {\scad}s can be further adapted by combining with other CAD techniques. We shall discuss here the interaction with the idea of truth table invariant CAD (TTICAD) introduced recently by \cite{BDEMW13}. 
We recall the definition of a TTICAD.

\begin{definition}[\cite{BDEMW13}]{\ }\\
Let $\Phi := \{ \phi_i\}_{i=1}^t$ be a list of quantifier-free formulae (QFFs).
We say a cylindrical algebraic decomposition $\mathcal{D}$ is a {\bf Truth Table Invariant} CAD for $\Phi$ (a {\bf TTICAD}) if the Boolean value of each $\phi_i$ is constant (either true or false) on each cell of $\mathcal{D}$.
\end{definition}

In \cite{BDEMW13} an algorithm was given to build TTICADs in the case where each $\phi_i$ contained an explicit equational constraint. This has recently been extended to allow for any $\phi_i$ \cite{England13b}.  The algorithms involve a new projection operator which encapsulates the interaction between the equational constraints of the $\phi_i$'s whilst ignoring interactions between polynomials that have no effect on the truth value of $\phi_i$.

Truth table invariance is very useful. Given a problem defined by a parent formula $\Phi$ built by a boolean combination of $\{ \phi_i\}_{i=1}^t$, a TTICAD is both sufficient to determine where $\Phi$ is true, and more efficient than any other projection operator.
Further, there are classes of problems for which a TTICAD is exactly the desired structure, such as the problem of decomposing a complex domain according to the branch cuts of multivariate functions \cite{BD02, PBD10, EBDW13}.

%While the earlier discussion of variety and layered {\scads} may have supposed sign-invariance, the theory transfers easily to truth-invariance.

In the case where there is a parent formula $\Phi$ and all the $\phi_i$ have their own equational constraint there exists a variety on which the solution rests.  It is defined by the product of the individual equational constraints (an implicit equational constraint for $\Phi$).  Hence in this case it makes sense to build a {\bf Variety Sub-TTICAD (V-sub-TTICAD)}.  
For simplicity, we assume that each equational constraint has factors which all have main variable $x_n$. % and they are not nullified on a low-dimensional cell.  
We can then use Algorithm \ref{alg:VarietySub-CAD} on $\Phi$ with the declared equational constraint the product of the individual ones.  {\tt ProjOp} should be the TTICAD projection operator from \cite{BDEMW13} applied to the sequence of bases of polynomials appearing in the formulae and {\tt GenerateStack} an algorithm which checks for the TTICAD well-orientedness properties.  A proof of the validity of Algorithm \ref{alg:VarietySub-CAD} using these sub-procedures to give a V-sub-TTICAD would follow analogously to the proof of Theorem \ref{thm:MCAD}.
Although the variety considered is an implicit equational constraint the TTICAD projection theory is more efficient that the equational constraint theory applied to this case (as described in \cite{BDEMW14}).  

Creating a \textbf{Layered Sub-TTICAD (L-sub-TTICAD)} is also possible, by applying Algorithm \ref{alg:layeredCAD} (or \ref{alg:recursivelayered}) with the TTICAD projection operator and a {\tt GenerateStack} algorithm which checks for its well-orientedness condition.  The validity is proven by Theorem \ref{thm:L1} (using results from \cite{BDEMW13}). 
 
Similarly, we can combine TTICAD with the ideas of Section \ref{subsec:LMCAD} to produce {\laymanttiscad}s as defined below.  

\begin{definition}{\ }\\
Let $\{ \phi_i\}_{i=1}^t$ be a list of QFFs with each $\phi_i$ having equational constraint $f_i$ whose factors all have main variable $x_n$. Let $1 \leq \ell \leq n$. A {\scad} of a TTICAD for  $\{ \phi_i\}_{i=1}^t$ containing all cells of dimension $n-1-i$ for $0 \leq i < \ell$ resting on the variety defined by $\prod_{i=1}^{t} f_i = 0$ is called a {\bf {\LayManttiScad} ({\LMttiscad})}.
\end{definition}

\begin{remark}{\ }\\
As with ML-{\scads}, a {\llaymanttiscad{$\ell$}} consists of the top $\ell$-layers of cells {\em on the variety $\prod_{i=1}^{t} f_i = 0$}. Again, this can be thought of as the intersection of an $(\ell+1)$-layered CAD of $\mathbb{R}^n$ with the variety (as the layer of $n$-dimensional cells is discarded when lifting to the variety).
\end{remark}

We can construct {\LMttiscad}s by using Algorithm \ref{alg:LVSubCAD} with the sub-algorithms implementing TTICAD.  
As noted, the correctness of the approaches from this subsection follow analogously to the proofs of Theorem \ref{thm:MCAD} - \ref{thm:L2} and \ref{thm:LV} (noting that the exceptional cases where part of a variety is nullified would have meant the input was not well-oriented for TTICAD and thus triggering an output of FAIL).  The example in Section \ref{SUBSEC:LMTTICASD} demonstrates the use of a {\LMttiscad} and the benefits of choosing to do so.  See the technical report \cite{WE13} for some further details and examples.

\section{Complexity analysis}
\label{sec:Complexity}

We provide a complexity analysis of the algorithms to compute sign-invariant variety {\scad}s and {\llaymanscad{1}}s in the case where the equational constraint has all factors with main variable $x_n$.  We need to study three parts of the complexity:
\begin{description}
  \item[Projection] The complexity of the equational constraint projection set needs to be analysed. In particular the number of polynomials, the maximum degree and size of their coefficients.
  \item[Calculation of $\bm{(n-1)}$ dimension CAD] These values then can be used to estimate the complexity of the $(n-1)$--dimensional CAD.
  \item[Lifting] Finally the complexity of the lifting stage can be combined with the previous step to describe the complexity of the variety CAD.
\end{description}
We recall the previous comprehensive work on CAD complexity by Collins \cite{Collins1975} and McCallum \cite{McCallum1993}.

We first standardise some notation for a CAD with respect to a set of polynomials $A$: let $n$ be the number of variables, $m$ the number of polynomials in $A$, $d$ the maximum degree in any variable of the polynomials in $A$, and $l$ the maximum norm length of the polynomials in $A$ (where the norm length, $|f|_1$, is the sum of the absolute values of the integer coefficients of a polynomial).

Let $A_1 := A$ and let $A_{i+1} = {\rm Proj}(A_i)$. In general the projection operator used will be clear: most of the following is with respect to Collins' projection operator and, therefore, is a `worst case scenario' compared to the improved operators of McCallum \cite{McCallum1998} and Brown \cite{Brown2001a} (which are subsets of the Collins operator). Let $m_k$ be the number of polynomials in $A_k$, $d_k$ the maximum degree of $A_k$, and $l_k$ the maximum norm length.

\subsection{Collins' algorithm}

Collins \cite{Collins1975} works through the original CAD algorithm in great detail to analyse the complexity, and this methodology is followed in \cite{McCallum1993}.  We recall some key results, noting they could all be uniformly improved with ideas  from papers such as \cite{Burr2013,Davenport1985} (but that is not the aim of this paper).

In the projection stage we can bound the properties of the projection sets as follows:
\[
m_k \leq (2d)^{3^k} m^{2^{k-1}}; \quad d_k \leq \frac{1}{2} (2d)^{2^{k-1}}; \quad l_k \leq (2d)^{2^k} l.
\]
By combining these bounds Collins shows the projection phase is dominated by
\[
  (2d)^{3^{n+1}} m^{2^n} l^2.
\]
The base case and lifting algorithm requires the isolation of real roots of univariate polynomials. Collins bounds this procedure as follows. Let $A$ be a set of univariate polynomials with degree bounded by $d$ and norm length bounded by $l$. Then for a given $f \in A$ with $\hat{d} := \deg(f)$ and $\hat{l} := |f|_1$ a lower bound on the distance between two roots is given by
\begin{equation}
\label{eq:rootseplowerbound}
\frac{1}{2} \left( \sqrt{e} \hat{d}^{\frac{3}{2}} \hat{l} \right)^{-\hat{d}}.
\end{equation}
Collins uses his analysis of Heindel's algorithm for real root isolation to show that isolating the roots is dominated by
\begin{equation}
\label{eq:Heindelalgorithm}
\hat{d}^8 + \hat{d}^7 \hat{l}^3.
\end{equation}
Therefore the operations needed to isolate all roots in $A$, with $m:=|A|$, is dominated by:
\begin{equation}
\label{eq:singlerootisol} 
md^8 + md^7l^3.
\end{equation}

For a given $h$, refining a root interval to $2^{-h}$ is dominated by $d^2h^3+d^2l^2h$. Collins multiplies all polynomials in $A$ and uses \eqref{eq:rootseplowerbound} to show that all roots are separated by:
\[
\delta := \frac{1}{2} \left( \sqrt{e} (md)^{\frac{3}{2}} l^m \right)^{-md}.
\]
If we take $h$ to be $\log(\delta)$ we can refine all intervals of the polynomials in 
\[
%md(d ^5l^3 + d^3l^3).   % DJW's version
md(d ^2 (m^2dl+md\log(md))^3 + md^3l^3) = O(m^7d^7l^3+m^2d^4l^3).
\]
Combining this with \eqref{eq:singlerootisol} tells us that the necessary refinement of intervals for all the polynomials is dominated by
\begin{equation}
md^8 + m^7d^7l^3. \label{eq:intervalrefinement}
\end{equation}
Combining \eqref{eq:intervalrefinement} with the size of the full projection set  concludes the base phase is dominated by
\[
(2d)^{3^{n+3}} m^{2^{n+2}} l^3.
\]

We also need to consider the polynomials involved in the lifting stage. This involves looking at the univariate polynomials created after substituting sample points, along with the polynomials required to define the algebraic extensions of $\mathbb{Q}$ that those sample points are contained in.

We follow \cite{Collins1975,McCallum1993} in using primitive elements to calculate the costs of operations, even though implementers are extremely unlikely to use them. 
For each sample point $\beta = (\beta_1,\ldots,\beta_k) \in \mathbb{R}^k$ there is a real algebraic number $\alpha \in \mathbb{R}$ such that $\mathbb{Q}(\beta_1,\ldots,\beta_k) = \mathbb{Q}(\alpha)$. Let $A_\alpha$ be the polynomial in $\mathbb{Q}[x]$ that, along with an isolating interval $I_\alpha$, defines $\alpha$ (so $A(\alpha) = 0$).  Let $d_k^*$ be the maximum degree of these $A_\alpha$, and $l_k^*$ the maximum norm length. Each coordinate $\beta_i$ is represented in $\mathbb{Q}(\alpha)$ by another polynomial.  Let $l_k'$ be the maximum norm length of these polynomials. Then
\[
d_k^* \leq (2d)^{2^{2n-1}}, \quad {\rm and} \quad  l_k^*, l_k' \leq (2d^2)^{2^{2n+3}} m^{2^{n+1}} l.
\]
Let $u_k$ be the number of univariate polynomials (after substitution) and $c_k$ the number of cells at level $k$. Then
\begin{equation}
u_k, c_k \leq 2^{2^n} \prod_{i=1}^n m_id_i, \quad {\rm and} \quad u_k,c_k \leq (2d)^{3^{n+1}} m^{2^n}. \label{eq:ukckCAD}
\end{equation}
Collins combines all these results to give a complexity bound for the full algorithm of
\begin{equation}\label{eq:CollinsComplexity}
(2d)^{4^{n+4}} m^{2^{n+6}} l^3.
\end{equation}

\subsection{McCallum's {\sc Cadmd} algorithm}

In \cite{McCallum1993}, McCallum gives a complexity bound, using the same methodology as \cite{Collins1975}, for his {\sc Cadmd} algorithm which produces a 1-layered CAD. Thanks to the avoidance of algebraic numbers (all sample points in a 1-layered CAD can be produced directly in $\mathbb{Q}$) the exponents are lower than in \eqref{eq:CollinsComplexity}. The complexity is dominated by
\begin{equation}\label{eq:McCallumComplexity}
(2d)^{3^{n+4}} m^{2^{n+4}} l^3.
\end{equation}

\subsection{Analysis of $P_E(A)$ and $(n-1)$-dimensional CADs}

To have an accurate complexity we must consider some properties of the projection set: the size, maximum degree, and maximum norm length. Let $E$ be the subset of $A$ containing the factors of the designated equational constraint.
Let $m_A := |A|$, $m_E := |E|$, $m_{A \setminus E} := |A \setminus E|$, and let $d_A$, $d_E$, $d_{A \setminus E}$, $l_A$, $l_E$, and $l_{A \setminus E}$ be defined similarly.

Since we are constructing a CAD with respect to equational constraints, we assume use of the projection operator, $P_E(A)$, defined by McCallum \cite{McCallum1999} as
\[
P_E(A) := P(E) \cup \{ {\rm res}_{x_n}(f,g) \mid f \in E, g \in A \setminus E \}
\]
where $P(E)$ is an application of the operator defined in \cite{McCallum1998} (giving the coefficients, discriminants and cross resultants of $E$). We will denote the resultant set, $P_E(A) \setminus P(E)$, by ${\rm ResSet_E(A)}$. 

The size of $P(E)$ is bounded by the number of coefficients ($m_Ed_E$), discriminants ($m_E$), and resultants ($\binom{m_E}{2} = \frac{m_E(m_E-1)}{2}$). Therefore we have
\begin{align}
\left|P_E(A)\right| \leq m_Ed_E + m_E + \frac{m_E(m_E-1)}{2} + m_Em_{A \setminus E} %\nonumber \\
&= \frac{m_E}{2} \left( 2d_E + 2m_{A \setminus E} + m_E -1 \right) \nonumber \\ 
&= \frac{m_E}{2} \left( 2d_E + m_{A \setminus E} + m_A -1 \right). \label{eq:pEAsize}
\end{align}

The maximum degree of $P_E(A)$ is the greater of the maximum degrees of $P(E)$ and the resultant set. We also have a bound on the degree of a resultant with respect to $x$:
\[
\deg ( {\rm res}_x(f,g)) \leq \left( \deg_x f + \deg_x g\right) \cdot ( \max(\deg_y f, \deg_y g)).
\]
Using our overall degree bounds gives 
\begin{align}
\max \deg ({\rm ResSet}_E(A)) \leq (d_E + d_{A \setminus E}) \cdot \max(d_E,d_{A \setminus E})% \nonumber\\
&= \max(d_E^2 + d_Ed_{A \setminus E}, d_{A \setminus E}^2 + d_Ed_{A \setminus E}) \nonumber \\
&\leq \max(2d_E^2,2d_{A \setminus E}^2).\label{eq:ResSetdegreebound}
\end{align}
We also have
\begin{equation}
\max \deg (P(E)) \leq \max (d_E, 2d_E^2,2d_E^2) = 2d_E^2. \label{eq:PEdegreebound}
\end{equation}
Combining \eqref{eq:ResSetdegreebound} and \eqref{eq:PEdegreebound} gives a degree bound for $P_E(A)$:
\begin{equation}
\max \deg (P_E(A)) \leq \max(2d_E^2, 2d_{A \setminus E}^2) \leq d_A^2. \label{eq:PEAdegreebound}
\end{equation}
Finally, if we denote the maximum norm length of $P_E(A)$ by $\overline{l}$, then we know $\overline{l} \leq l_2$ (since $P_E(A) \leq {\rm Proj}(A)$) and so
\begin{equation}
\overline{l} \leq l_2 \leq (2d_A)^{2^2} l_A = 16d_A^4l_A. \label{eq:PEAnormlength}
\end{equation}

Substituting the bounds from \eqref{eq:pEAsize}, \eqref{eq:PEAdegreebound} and \eqref{eq:PEAnormlength} into \eqref{eq:CollinsComplexity} and \eqref{eq:McCallumComplexity} gives an estimate on the complexity of a complete $(n-1)$-dimensional $P_E(A)$-invariant Collins CAD dominated by
\begin{align}
&\leq \left( 2 \cdot 2 d_A^2\right)^{2^{2(n-1)+8}} \left( \frac{m_E(2d_E + m_A + m_{A\setminus E}-1)}{2} \right)^{2^{n-1+6}} \left( 16 d_A^4 l_A \right)^3 \nonumber \\
&\leq 16^3\left( 4 d_A^2\right)^{2^{2n+6}} \left( \frac{m_E(2d_E + 2m_A -1)}{2} \right)^{2^{n+5}} l_A^3 d_A^{12}. \label{eq:n-1Collins}
\end{align}
Similarly, the complexity of a 1-layered $(n-1)$-dimensional $P_E(A)$-invariant CAD is dominated by
\begin{align}
&\leq \left( 2 \cdot 2 d_A^2\right)^{3^{n-1+4}} \left( \frac{m_E(2d_E + m_A + m_{A\setminus E}-1)}{2} \right)^{2^{n-1+4}} \left( 16 d_A^4 l_A \right)^3 \nonumber \\
&\leq 16^3\left( 4 d_A^2\right)^{3^{n+3}} \left( \frac{m_E(2d_E + 2m_A -1)}{2} \right)^{2^{n+3}} l_A^3 d_A^{12}. \label{eq:n-1McCallum}
\end{align}

\subsection{Overall complexities}

So far, we have computed the complexities of the $(n-1)$-dimensional CADs. The final step is to lift over these cells with respect to the equational constraints. From \eqref{eq:ukckCAD} we can bound the number of univariate polynomials and so know from \eqref{eq:intervalrefinement} the isolations will be dominated by:
\begin{equation}
\left(2d_E\right)^{3^{n+1}} m_E^{2^n} d_E^8 + \left( \left(2d_E\right)^{3^{n+1}} m_E^{2^n}\right)^7 d_E^7 l_E^3. \label{eq:eqconstcomplexity}
\end{equation}

\iffalse
\begin{remark}{\ }\\
\TODO{New Complexity stuff} We could also use the improvements of \eqref{eq:swedish}, \eqref{eq:swedishandBurr}, and \eqref{eq:straightforwardBurr} in this final lifting stage but for the sake of consistency use \eqref{eq:intervalrefinement}.
\end{remark}
\fi

By combining \eqref{eq:eqconstcomplexity} with \eqref{eq:n-1Collins} and \eqref{eq:n-1McCallum} we are now in a position to describe the overall complexities of our algorithms. Note that \eqref{eq:eqconstcomplexity} will be a large overestimation for the 1-layered case.  % but will be sufficient for our purposes.

\begin{theorem}{\ }
\label{thm:complexityresults} \\
The complexity for computing a variety {\scad} using Collins' algorithm is dominated by:
\begin{multline}
2^{12}\left( 2^2 d_A^2\right)^{{\bf 4^{n+3}}} \left( \frac{m_E(2d_E + 2m_A -1)}{2} \right)^{{\bf 2^{n+5}}} l_A^3 d_A^{12} + \\ \left(2d_E\right)^{3^{n+1}} m_E^{2^n} d_E^8 + \left( \left(2d_E\right)^{3^{n+1}} m_E^{2^n}\right)^7 d_E^7 l_E^3. \label{eq:totalcomplexityCollins}
\end{multline}
The complexity for computing a {\llaymanscad{1}} is dominated by:
\begin{multline}
2^{12}\left( 2^2 d_A^2\right)^{{\bf 3^{n+3}}} \left( \frac{m_E(2d_E + 2m_A -1)}{2} \right)^{{\bf 2^{n+3}}} l_A^3 d_A^{12} + \\ \left(2d_E\right)^{3^{n+1}} m_E^{2^n} d_E^8 + \left( \left(2d_E\right)^{3^{n+1}} m_E^{2^n}\right)^7 d_E^7 l_E^3. \label{eq:totalcomplexityMcCallum}
\end{multline}
%where we have used bold face to indicate the key differences in the exponents.
\end{theorem}

%\subsection{Comparison of complexities}

In Theorem \ref{thm:complexityresults} we have emboldened the exponents to highlight the difference between \eqref{eq:totalcomplexityCollins} and \eqref{eq:totalcomplexityMcCallum}.
To help visualise the comparison we have plotted the double logarithm of the complexities against $n$ in Figure \ref{fig:doublelogarithm} (for some specific parameter values).
% $d_A = 3$, $d_E=2$, $m_A=3$, $m_E = 1$, $m_{A \setminus E}=2$, $l_A=2$, $l_E=2$.  
The diagram shows the drop in the constant in the exponents of \eqref{eq:totalcomplexityCollins} and \eqref{eq:totalcomplexityMcCallum}, whilst the scaling factor of the exponent remains the same between variety and non-variety versions of each algorithm.

\begin{figure}
\begin{center}
\includegraphics[width=0.45\textwidth]{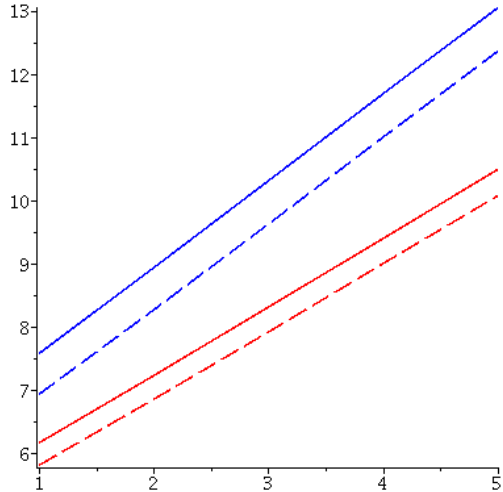}
\end{center}
\caption{A plot of $n$ (horizontal axis) against the double logarithm of the complexities of algorithms (vertical axis).  The complexities were evaluated with parameter choices $d_A = 3$, $d_E=2$, $m_A=3$, $m_E = 1$, $m_{A \setminus E}=2$, $l_A=2$, $l_E=2$. From top to bottom: CAD, variety {\scad}, 1-layered {\scad} and \llaymanscad{1}.  }
\label{fig:doublelogarithm}
\end{figure}

\section{Examples and Implementation}
\label{sec:ExImp}

We provide some case studies showing the benefit of our new algorithms, which have all been implemented in the {\sc Maple} package {\tt ProjectionCAD} \cite{England13a, England13b, WE13}.

We also compare to some competing CAD implementations: the CAD procedures in {\sc Maple}'s \texttt{RegularChains} Library, the command line program {\sc Qepcad} \cite{Brown2003b}, and the algorithm in {\sc Mathematica} \cite{Strzebonski10} (which produces a cylindrical algebraic formula).  We tested two \texttt{RegularChains} routines: the one following \cite{CMXY09} (distributed with \textsc{Maple}), and the one following \cite{CM12b} which can make use of equational constraints.  
With \textsc{Qepcad} we ran it first on with its default settings (implementing \cite{McCallum1998}) and also when an equational constraint is designated (where it follows \cite{McCallum1999}).  

\textsc{Qepcad} also has the option \texttt{measure-zero-error} which produces a CAD with only the full-dimensional cells of the free variable space guaranteed to satisfy the invariance condition (see \cite{Brown2003b}).  Although this is a CAD rather than a sub-CAD of the free variable space it is only sufficient to ensure the full dimensional solutions are correct, like a 1-layered {\scad}, meaning the solutions have an \emph{error of measure zero in the free-variable space}.  
For the three case studies below this command was not of use because the problems were unquantified and each have an equational constraint, meaning their formulae can only be satisfied on cells of less than full-dimension (in the free variable space).  So although we could produce an output from \textsc{Qepcad} using \texttt{measure-zero-error}, the only cells guaranteed to be correct are not of interest.  This contrasts with a 1-layered variety sub-CAD which does provide valid solutions.  Here the layer is with respect to the variety, not free variable space, and so the solutions provided have an \emph{error of measure zero in the solution space}.

Experiments were run on a Linux desktop (3.1GHz Intel processor, 8.0Gb total memory).  Tests for {\sc Mathematica} used V9 and for {\sc Qepcad} used \textsc{Qepcad-B} 1.69 with the options {\tt +N500000000} and {\tt +L200000}  (initialisation times included).  The tests in \textsc{Maple} used the development version (similar to {\sc Maple} 18) in command line interface, using the development version of the \texttt{RegularChains} Library\footnote{Available from \texttt{www.regularchains.org}.}.

\subsection{Example: Making use of a 1-layered variety {\scad}}
\label{SUBSEC:LMCASD}

Assume variable ordering $x \succ y \succ z$ and consider the following formula involving 3 random polynomials of degree 2 (generated using {\sc Maple}'s {\tt randpoly} function) which are plotted in Figure \ref{fig:Example1diagram}:
\begin{align*}
&\Phi := -50 x y + 56 y z + 41 z^2 + 67 x - 55 y - 21 = 0 \\
& \quad \land \quad 36 x y + 76 x z - 58 y z + 69 z^2 + 75 y + 27 > 0 \\ 
& \quad \land \quad -55 x^2  + 10 x y - 88 x + 80 y + z - 39 > 0.
\end{align*}
We wish to describe the regions of $\mathbb{R}^3$ in which $\Phi$ is satisfied. Figure \ref{fig:Example1diagram} shows there are multiple intersections between the two non-equational constraints away from the variety defined by the equational constraint. This suggests that V-{\scad} could be beneficial.  Further, if only generic solution sets are of interest then a {\llaymanscad{1}} could offer further savings.

\begin{figure}[t]
\[
\includegraphics[width=8cm]{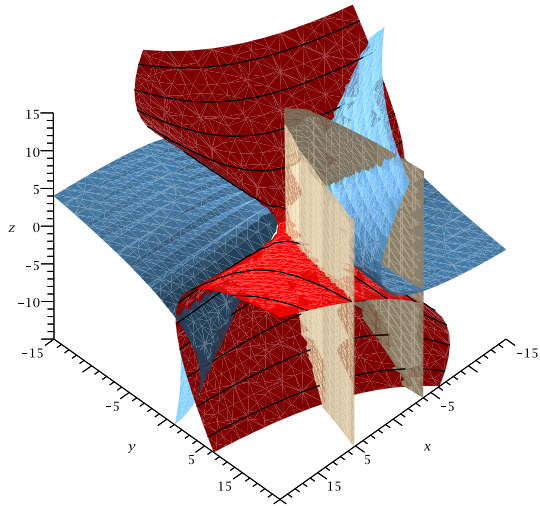}
\]
\caption{Intersection of the three surfaces from Section \ref{SUBSEC:LMCASD}.  The red surface (the darkest in black and white) is the equational constraint.}
\label{fig:Example1diagram}
\end{figure}

To show the relative benefits of the individual theories we solved the problem by constructing: a full sign-invariant CAD for the three polynomials (following \cite{McCallum1998}); a CAD invariant with respect to the equational constraint (following \cite{McCallum1999}); a variety {\scad} (following Algorithm \ref{alg:VarietySub-CAD} implementing sub-procedures from \cite{McCallum1999}) and layered variety {\scads} (following Algorithm \ref{alg:LVSubCAD} implementing sub-procedures from \cite{McCallum1999}).
\begin{description}  
	\item[Full sign-invariant CAD] 17,047 cells, 178.277 seconds.
	\item[CAD invariant with respect to EC] 1315 cells, 11.520 seconds.
	\item[Variety {\scad}] 422 cells, 10.723 seconds.
	\item[{\lLayManScad{2}}] 348 cells, 7.149 seconds.
	\item[{\lLayManScad{1}}] 138 cells, 0.475 seconds.
\end{description}
We see that making use of the equational constraint in the projection stage dramatically reduces the computation involved.  A variety {\scad} further reduces the size of the output which will lead to time savings on any future work.  
The variety {\scad} is sufficient to describe exactly where $\Phi$ is true, but if we are only concerned with the generic solution sets then a {\llaymanscad{1}} can be used to achieve further time-savings.  Note that this means an output of over 17,000 cells can be replaced with one of only 138.  If solutions of lower dimension are also needed then the {\llaymanscad{2}} (348 cells) or the complete variety {\scad} (422 cells) also offer great savings.

We now consider competing CAD implementations:
\begin{description}  
	\item[Maple: Algorithm from \cite{CMXY09}] 9841 cells, 112.460 seconds.
	\item[Maple: Algorithm from \cite{CM12b}] 559 cells, 1.999 seconds.	
    \item[Qepcad] 17,047 cells, 385.679 seconds.
    \item[Qepcad EC] 5271 cells, 26.614 seconds.
%    \item[Qepcad (error-measure-zero)] 3082 cells, 4.879 seconds.
%    \item[Qepcad EC (error-measure-zero)] 822 cells, 4.820 seconds.
    \item[Mathematica] 0.533 seconds
\end{description}
We see that those algorithms taking advantage of the equational constraint offer smaller quicker CADs.  {\sc Qepcad} does comparatively worse here (perhaps due to the lack of improved lifting described in \cite{BDEMW14}).

If we did use \textsc{Qepcad} with its \texttt{measure-zero-error} option here it would return (in under 5 seconds) a CAD with 822 cells, but $\Phi$ would not be satisfied on any of them (as expected) and hence \textsc{Qepcad} gives an equivalent quantifier free formula \texttt{False}.   For the \textsc{Qepcad} outputs above $\Phi$ is only valid on a small fraction of cells (290/17047 and 106/5271).
%, 61/3082 and 22/822).  
The false cells that are constructed (and on which $\Phi$ is evaluated) will make a significant contribution to the computation time.
We can evaluate $\Phi$ on all the cells produced in the {\lLMscad{1}} almost instantly to find that there are 36 of 138 cells on which $\Phi$ is satisfied (the equation is by definition satisfied for all of them but the truth of the other constraints varies).
 
%Whilst not defining all possible solutions, those missed are of measure zero in the solution set.

{\sc Mathematica} does not produce CAD cells, and so we cannot compare cell counts.  Instead it produces a cylindrical formula, from which the solutions can be derived.  This is done very quickly, (probably due to the symbolic-numeric techniques discussed in \cite{Strzebonski06}).

\bigskip

Consider a general problem of the form $ f = 0 \land \Psi(g_i)$ where $f=0$ defines a variety of {\em real\/} co-dimension 1 and $\Psi$ is a quantifier-free formula involving only $f=0$ and strict inequalities.
We expect that a {\llaymanscad{1}} will usually be sufficient to describe the generic solutions.
%those differing only by a set of measure zero from the full-solution set) 
We may find that there are no solutions of full dimension on the variety, and can use the recursive approach for layered CAD to incrementally build extra layers as required.

\subsection{Example: Making use of 1-layered variety {\ttiscad}}
\label{SUBSEC:LMTTICASD}

We now consider a problem suitable for both our new {\scad} approaches and TTICAD.
Define the following quantifier free formulae:
\begin{align*}
&\varphi_1 :=  x^2+y^2+z^2=1 \land xy+yz+zx<1 \land x^3 - y^3 - z^3 < 0, \\
&\varphi_2 :=  (x-1)^2+(y-1)^2+(z-1)^2 = 1 \land (x-1)(y-1)+(y-1)(z-1) \nonumber \\
&\qquad \qquad \quad + (z-1)(x-1) < 1 \land (x-1)^3-(y-1)^3-(z-1)^3 <0.
\end{align*}
The surfaces defined by the polynomials in $\varphi_1$ are shown in Figure \ref{fig:Example2diagram1}, while those in $\varphi_2$ are the same but shifted.
Assume the variable ordering $x \succ y \succ z$ and consider the problem of finding all regions of $\mathbb{R}^3$ satisfying 
$\Phi := \varphi_1 \lor \varphi_2.$
It would be na\"ive to tackle this problem with a sign-invariant CAD for the 6 polynomials in $\Phi$, (the default CAD in \textsc{Maple 16} could not produce one when left overnight while {\sc Qepcad} gives a ``prime list exhausted'' error (a memory constraint) after two hours).  
%If only generic solutions are of interest then {\sc Qepcad} can produce a CAD with error of measure zero in the solutions in less than 5 seconds, with 3082 cells, (similarly a 1-layered {\scad} can be produced with xyz cells.

Instead we should make use of the equations present in the formulae.  We can do this with {\sc Qepcad} by declaring the implicit equational constraint (the product of the two equations) to give a CAD with 6165 cells in 35,304 seconds (around 10 hours).  \textsc{Mathematica} can produce a cylindrical formula in  1.805 seconds. 
This problem is well suited for TTICAD and applying the {\tt ProjectionCAD} implementation on the two formulae $\varphi_1$ and $\varphi_2$ produces 4861 cells in 170.515 seconds.  A substantially lower cell count than \textsc{Qepcad} is achieved because the TTICAD projection set is smaller than the one using the implicit equational constraint (see \cite{BDEMW13, BDEMW14}).

\begin{figure}
\[
\includegraphics[width=7cm]{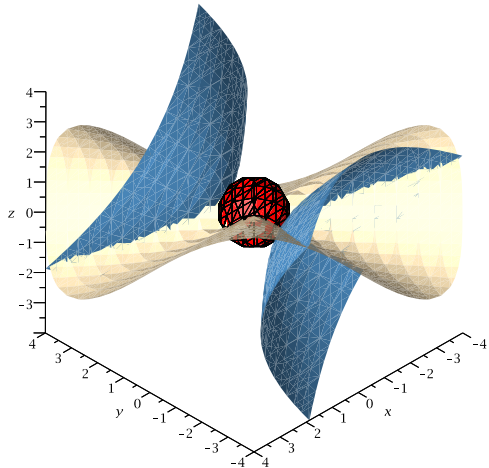}
\]
\caption{Intersection of the surfaces from $\varphi_1$ in Section \ref{SUBSEC:LMTTICASD} -- the sphere defined by the equational constraint is in red (the darkest surface in black and white). }
%\TODO{Captions cannot refer to colours}
\label{fig:Example2diagram1}
\end{figure}

We now consider how the TTICAD can be improved upon using the new theory.  Each formula $\varphi_i$ contains an equational constraint and so the formula $\Phi$ is only true on the variety defined by their product.  Suppose further that we only want to obtain the generic solutions.  Then we may apply the TTICAD operator to $\varphi_1$ and $\varphi_2$ and construct a 1-layered {\scad} of $\mathbb{R}^2$ with respect to this projection set. This takes 0.947 seconds and produces 249 cells in $\mathbb{R}^2$. We then lift with respect to both of the equational constraints onto the variety defined by their product. This takes a further 1.191 seconds and produces 528 2-dimensional cells on the 2-dimensional variety in $\mathbb{R}^3$. 
So the 1-layered variety {\ttiscad} saves 88\% of the cells and 99\% of the computation time of the TTICAD. 

The {\llaymanttiscad{1}} obtains all the cells of full dimension (with respect to the variety) on which $\Phi$ is true, and so is sufficient up to a set of measure zero.  %\textsc{Qepcad} can also produce an output with error of measure zero with has 1790 cells and takes 4.82 seconds (the higher cell count due to the less efficient projection operator and not restricting to the variety.)  BUT THEY ARE ALL FALSE.
If solutions of lower dimension are needed then the theory in this paper allows for a {\llaymanttiscad{2}} or the full variety {\ttiscad} which contain 1514 cells and 1976 cells, respectively.  The former takes under a minute 
%47.629 seconds
and the latter just less than three to compute.
%178.196
So we see that the variety {\ttiscad} takes about the same time as a TTICAD, but gives a cell saving, while the layered variety {\ttiscad}s also offer time savings.  Note that over half the cells in the complete TTICAD do not lie on either of the varieties defined by the equational constraints.

To magnify the issues consider a third formula (a different shift of the original surfaces)
\begin{align*}
\varphi_3 &:= (x+1)^2+(y+1)^2+(z+1)^2 = 1 \land (x+1)(y+1)+(y+1)(z+1) \nonumber \\ &\qquad \qquad +(z+1)(x+1) < 1 \land (x+1)^3-(y+1)^3-(z+1)^3 <0,
\end{align*}
and a new overall formula, 
%\[
$\Phi^* := \varphi_1 \lor \varphi_2 \lor \varphi_3.$
%\]
%The surfaces involved in $\Phi^*$ are shown in Figure \ref{fig:Example2diagram2}

%\begin{figure}
%\[
%\includegraphics[scale=0.25]{example2diagram2}
%\]
%\caption{Intersection of the surfaces from $\Phi^*$ in Section \ref{SUBSEC:LMTTICASD} -- the spheres defined by the equational constraints are given in red.  }\label{fig:Example2diagram2}
%\end{figure}

Given the above results for $\Phi$ we do not attempt to solve this problem without utilising the equational constraints.  This time to build a 1-layered variety sub-TTICAD takes 5.003 seconds and produces 1104 cells.  
%We repeat the approach used for $\Phi$.  This time the 1-layered {\scad} of $\mathbb{R}^2$ takes 3200.389 seconds (just under an hour) and produces 488 cells. It takes a further 7.250 seconds to construct the {\llaymanttiscad{1}}, resulting in 1096 cells.  
The full TTICAD takes 432.210 seconds and produces 10063 cells; around 10 times as many. Although the TTICAD will contain all valid cells for $\Phi$, the {\llaymanttiscad{1}} will provide descriptions of the generic families of solutions.
If solutions of lower dimensions are needed then a {\llaymanttiscad{2}} or the complete variety {\ttiscad} would both be preferable to the full TTICAD.  The former contains 3166 cells and took 145.898 seconds, and the latter contains 4130 cells and took 429.083 seconds.  %This again highlights the fact that over half the cells produced in the TTICAD do not satisfy any of the equational constraints and are thus unnecessary to describe the solution.

%Mathematica takes 3.407482
%Qepcad emz: 95744, 13.064 (10.313)
%Qepcad EC emx: 7951, 5.875 (3.008)

\subsection{Example: A ``Piano Movers'' problem}
\label{SUBSEC:Piano}

An application of CAD of great interest is motion planning. Given a semi-algebraic object, an initial and desired position, and semi-algebraic obstacles, a CAD can be constructed of the valid configuration space of the object. The connectivity of this space can then be used to determine if a feasible path from the initial position of the object to the desired endpoint is possible \cite{SS83II}.

A well-studied problem in this area is the movement of a ladder through a right-angled corridor, as proposed in \cite{Davenport1986}.  The problem consists of a ladder of length 3 inside a right-angled corridor of width 1 with the aim of moving from from position 1 to position 2 in Figure \ref{fig:Piano}. 

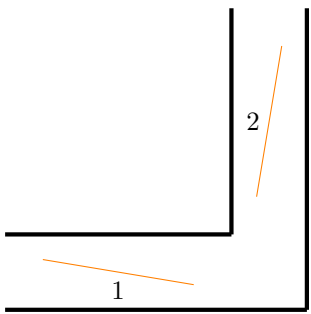
\begin{figure}[ht]
\centering
\begin{tikzpicture}
\draw[ultra thick] (0,0)--(-4,0);
\draw[ultra thick] (0,0)--(0,4);
\draw[ultra thick] (-1,1)--(-4,1);
\draw[ultra thick] (-1,1)--(-1,4);

\draw[orange] (-2/3,1.5)--(-1/3,3.5);
\draw[orange] (-1.5,1/3)--(-3.5,2/3);

\node [below] at (-2.5,0.5) {1};
\node [left] at (-0.5,2.5) {2};

\end{tikzpicture}
\caption{The piano movers problem considered in \cite{Davenport1986}}
\label{fig:Piano}
\end{figure}

Although in 2-dimensional real space, the problem lies in a 4-dimensional configuration space (as usually the problem is described using 4 variables specifying the endpoints of the ladder) making it far more difficult to describe using CAD.  The original formulation given in \cite{Davenport1986} proves particularly difficult but recently in \cite{WDEB13} a reformulation beneficial to CAD was given.  A CAD of the configuration space was built using  {\sc Qepcad} with 285,419 cells in around 5 hours.  This used the equational constraint \cite{McCallum1999} and partial CAD techniques \cite{CH91} (without these it increases to 1,691,473 cells and over 24 hours computation time).

A representation of the 2-dimensional CAD produced on route to the full 4-dimensional CAD is given in Figure \ref{fig:PianoCAD}.  {\sc Mathematica} can produce a cylindrical formula in 558.721 seconds, but for this application such a formula is not sufficient to deduce paths (since that will require knowledge of cell adjacencies and thus the boundary cells which can not always be inferred from the formula).

In the formulation the length of the ladder is an explicit equational constraint,
%Denoting the endpoints of the ladder as $(x,y)$ and $(x',y')$, the length of the ladder becomes the explicit equational constraint:
%\begin{equation*}
%(x-x')^2 + (y-y')^2 = 9.
%\end{equation*}
and so the problem is suited to treatment with a V-{\scad}. Indeed, we can see from Figure \ref{fig:PianoCAD} that there is a great deal of information computed  which is of no use in describing the suitable paths. Ideally we would restrict to a CAD or {\scad} of just the corridor highlighted.

Within this three-dimensional variety (embedded within $\mathbb{R}^4$) the important cells are those that are three-dimensional, since cells of lesser dimension correspond to physically infeasible situations (i.e. one dimensional subspaces of $\mathbb{R}^2$). Therefore a {\LMscad} would give an overview of the problem.

Constructing the 1-layered {\scad} of $\mathbb{R}^3$ produces 64,764 cells in around 124.22 seconds, and lifting to the variety takes a further 196.672 seconds, producing 101,924 cells, so offering substantial savings. 
%\textsc{Qepcad} can produce a CAD with error of measure zero (utilising the equational constraint) with 44,971 cells in 5.642 seconds. BUT ALL FALSE. 
%For this example, \textsc{Qepcad} is able to obtain more savings than our implementation as it used partial CAD techniques \cite{CH91} in free variable space (relevant since several of the constraints do not contain all variables).
Since this example has several constraints which do not contain all variables it is likely that partial CAD techniques (not yet implemented in our \textsc{Maple} package) would offer further savings.  %Nevertheless, we see that the variety and layered sub-CAD theory also offer substantial savings.  
Adjacency information will have to be computed for a full solution to the motion planning problem (which may require a {\llaymanscad{2}}).
%, not available in \textsc{Qepcad}). 
% 2-layered one ran out of memory.

\begin{figure}
\centering
\includegraphics[scale=0.7]{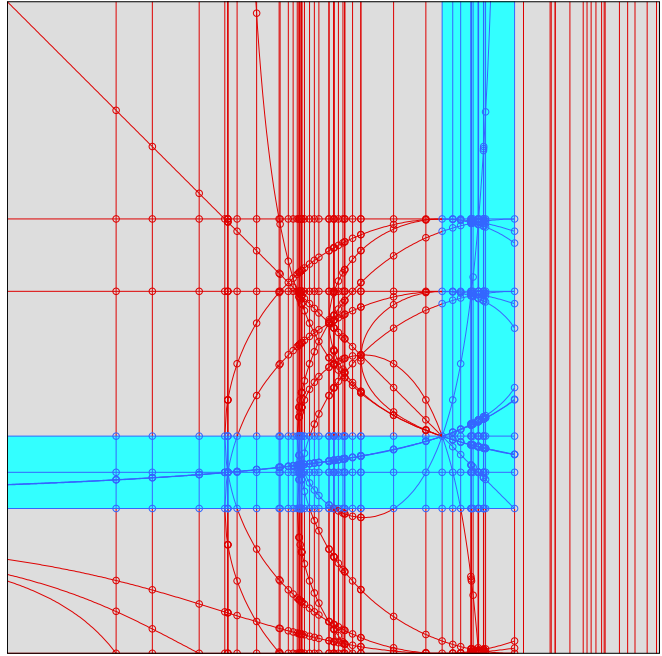}
\caption{A representation of the 2-dimensional induced CAD produced for the new formulation of the piano movers problem in \cite{WDEB13}.}
\label{fig:PianoCAD}
\end{figure}

\section{Conclusions and further work}
\label{sec:Conc}

We have formalised the idea of a cylindrical algebraic sub-decomposition. Whilst a simple idea, it can be hugely powerful and the examples presented show that massive cell reductions are possible.  In some cases time reductions are also available, and since most problems using CAD will require some further computation on the cells (such as polynomial evaluation) more time savings will follow and applications involving complicated calculation on the cells will benefit even more.  For example, the calculation of adjacency information for use in motion planning or the evaluation of multi-valued functions at (possibly algebraic) sample points for branch cut analysis. 

We provided examples of {\scads} in the literature, along with two new approaches (and algorithms to produce them): variety {\scads} and layered {\scads}. We find that their individual savings may be magnified by combining them with each other and the recent theory of truth table invariance.
The savings are large enough to tackle problems previously infeasible. 

There is great scope for future work, with some important questions as follows:
\begin{itemize}
\item Can we identify further classes of problems where the various types of {\scad} are sufficient?
\item What is the best way to build V-{\scads} for lower dimensional varieties?
\item How can we best adapt existing techniques (such as partial CAD) to output {\scads}? 
\item Can we develop heuristics (or adapt existing ones \cite{DSS04,BDEW13}) for when to use different {\scad} approaches?
\item Can we keep track of where cells arise when constructing a variety {\ttiscad} so that over each cell we lift only to the varieties for relevant ${\phi_i}$? This can be thought of as an analogue of partial CAD for TTICAD. This may alter the output significantly (it may not be a {\scad} of a CAD that can be constructed by current technology) but could allow for even smaller output for suitable problems. 
\item Can we parallelise the algorithms? The idea mentioned in \cite{McCallum1997} of lifting over sets of cells independently could be generalised.  (Preprocessing CAD problems to allow for parallelisation was discussed in \cite{MD13} but this involved only the boolean logic of the problem). 
% whilst trying to solve a motion planning problem, $\varphi$, in the plane the author identifies a subset of cells in the decomposition of $\mathbb{R}^1$ for which any valid cell for $\varphi$ must lie over. The author then lifts over these cells in groups of 5 cells (a useful technique for parallelisation of the lifting stage) to produce a {\scad} for the problem.
\end{itemize}
There are also interesting questions around which properties of CADs transfer over to their {\scads}. For example, any existing adjacency algorithms require CADs to be particularly `well-behaved', and it may be possible to avoid problematic cells through {\scads}, extending the use of such algorithms. Also, well-orientedness conditions for CAD algorithms \cite{McCallum1998,Brown2001a} may be failed for the CAD but not the {\scads}. These ideas need further investigation.

An overarching aim is to develop a general CAD framework to identify when applying each technique is appropriate, automatically combining appropriate methods when possible and making choices automatically when required based on heuristic information.  This would identify, for a given problem, an efficient way to produce a $\phi$-sufficient {\scad} and thus describe the solution set.

\section*{Acknowledgements}
This work was supported by the EPSRC grant: EP/J003247/1. The authors would also like to thank Professor Gregory Sankaran for his thoughts and feedback on the topic, and Professor Scott McCallum for many stimulating conversations on TTICAD. Finally, they would like to thank the anonymous referees for helpful comments which improved the paper.

\bibliographystyle{plain}
%\bibliography{DJW}
\bibliography{CAD}

\end{document}